%% file: ex_article.tex
\documentclass{siamart1116}
\usepackage{amsmath, amssymb, bm, mathrsfs}
\usepackage{cite, color, graphicx,subcaption}

\input{ex_shared}

\ifpdf
\hypersetup{
  pdftitle={\TheTitle},
  pdfauthor={\TheAuthors}
}
\fi

\newtheorem{assumption}[theorem]{Assumption}
\newtheorem{example}[theorem]{Example}
\newtheorem{remark}[theorem]{Remark}

\begin{document}

\maketitle

\begin{abstract}
This paper addresses quantized output feedback stabilization under Denial-of-Service (DoS) attacks.
First, assuming  that the duration and frequency of DoS attacks are 
averagely bounded and that an initial bound of the plant state is known,
we propose an output encoding scheme that achieves exponential convergence with finite data rates.
Next we show that a suitable state transformation allows us to remove the assumption on
the DoS frequency.
Finally, we discuss the derivation of state bounds under DoS attacks and obtain
sufficient conditions on the bounds of DoS duration and frequency for
achieving Lyapunov stability of the closed-loop system.
\end{abstract}


\section{Introduction}
Recent advances in computer and communication technology 
contribute to the efficiency of data transmission in control systems.
However, control systems become also vulnerable to cyber attacks.
For instance, it was reported that 
attackers can adversarially control cars \cite{Checkoway2011} 
and unmanned aerial vehicles \cite{Kerns2014}.
Malicious attacks are a major concern for the deployment of networked control systems, and
enhancing the resilience to cyber attacks is an important issue.

There are many possible attacks for control systems.
Recent results such as \cite{Fawzi2014, Chong2015, Shoukry2016} focus on
the scenario where measurement data obtained from some sensors
can be manipulated by malicious attackers.
Another line of research \cite{Mo2009, Zhu2014} investigates control under replay attacks, which maliciously
repeat transmitted data.
Denial-of-Service (DoS) attacks destroy the data 
availability by inducing packet losses.
DoS attacks are launched by malicious routers \cite{Awerbuch2008} and jammers \cite{Pelechrinis2011},
which can be set up without detailed knowledge on the structure of targeted systems.
Hence, even attackers with little information on control systems 
can create a security threat by DoS attacks.

In this paper, we consider
networked control systems in which the plant output is sent 
through a communication channel and DoS attacks are launched to block the transmission of
the output data over this channel. 
Probabilistic models such as the Bernoulli model has been used for nonmalicious packet losses caused by
network traffic congestion
and packet transmission failures; see the survey papers \cite{Hespanha2007,Zhang2013}.
However, attackers may not launch DoS attacks based on such probabilistic models.
The effect of DoS attacks has been recently investigated
in several studies \cite{
	Amin2009,
	Bhattacharya2013,
	Liu2014JA,
	Ding2017, Chen2018, 
	Persis2015,
	Persis2016,
	Feng2017,
	Cetinkaya2017,
	Senejohnny2017,
	Kikuchi2017,
	Cetinkaya2018,Lu2018, Feng2018}.
To deal with the uncertainty of DoS,
the previous studies \cite{
	Persis2015,
	Persis2016,
	Feng2017,
	Cetinkaya2017,
	Senejohnny2017,
	Kikuchi2017,
	Cetinkaya2018,Lu2018, Feng2018} characterized DoS attacks 
by the average duration and frequency of packet losses.

Data transmission through digital channels
requires signal quantization.
Although a plenty of communication bandwidth is available in modern
applications, many devices 
compete for this bandwidth in complex systems.
Moreover, it is theoretically interesting to solve
the problem of how much information is needed to achieve a given control objective.
From this point of view, data rate limitations for stabilization have been extensively studied;
see the survey papers \cite{Nair2007, Ishii2012} for details.
The so-called  zooming-in and zooming-out method developed in \cite{Brockett2000}
also yields a quantizer that achieves asymptotic stabilization with finite-data rates.
This method was first applied to linear time-invariant systems and then
was extended to a wide class of systems such as nonlinear systems 
\cite{Liberzon2003Automatica, Liberzon2005} and 
switched systems \cite{Liberzon2014, Wakaiki2017TAC}.

Despite the above active research on control problems with limited information,
quantized control under cyber attacks
does not seem to have received much attention so far.
In this paper, we extend the zooming-in and zooming-out method
to achieve output feedback stabilization under DoS attacks. 
Our objective is to develop output encoding schemes that 
guarantee closed-loop stability under DoS attacks.
The proposed encoding schemes generally require more than minimal data rates for stabilization but
relatively modest computational resources of the coders.
In contrast, data rate limitations for state feedback stabilization
under DoS attacks have been recently studied
in \cite{Feng2018}.
The authors of \cite{Chen2018} have proposed a design method of event-triggered 
controllers for stabilization under quantization and DoS attacks. However, 
static logarithmic quantizers with infinitely many quantization levels are used in \cite{Chen2018},
which
would remove most of the difficulties arising from quantization in our problem formulation. 

First, we assume that an initial bound of the plant state 
is given and design an output encoding scheme that achieves exponential 
convergence with finite data rates in the presence of DoS.
The difficulty here is to switch an update rule of the coders depending on DoS.
In the absence of DoS, the coders can decrease their quantization ranges and
make quantization errors small, by
using the plant model and the transmitted measurements. However,
if DoS attacks are launched, then
the decoder at the controller side cannot receive the measurements.
As a result, the worst-case estimation error of the plant output, which is used for quantization,
becomes large.
Therefore, the coders should increase their quantization range
so that the plant output can be captured in the quantization region.
This switching of the update rule of the coders makes it difficult to
analyze the stability of the closed-loop system.

We adopt a general model that constrains DoS attacks
only in terms of duration and frequency,  as in \cite{
	Persis2015,
	Persis2016,
	Feng2017,
	Cetinkaya2017,
	Senejohnny2017,
	Kikuchi2017, Cetinkaya2018, Lu2018, Feng2018}. 
In particular,
the assumption we make for DoS attacks is that 
their duration and frequency  are averagely bounded.
Hence we can deal with a wide class of packet losses.
We first propose an encoding scheme that generally needs the assumption 
both on DoS duration and frequency.
Next we show that the frequency condition can be removed, by
applying a suitable state transformation.
An invertible matrix for the state transformation is a design parameter, and
we can choose it in various ways.
In the section of a numerical example,
this matrix is chosen so that the closed-loop system allows longer DoS duration
under low DoS frequency.

Next, we develop methods to derive initial state bounds under DoS attacks.
In the absence of packet losses \cite{Liberzon2003}, 
state bounds can be obtained from consecutive output data.
In our setting, output data may not be received consecutively due to DoS attacks. Hence
we need to construct state bounds from intermittent output data.
In the case without DoS, 
it is easy to find state bounds from
finitely many measurements.
The difficulty of the case with DoS is that we may not obtain a state bound 
using even an infinite number of intermittent measurements. This is because there exist time-steps at  which the output  
does not
contribute to the construction of state bounds.
This problem is related to basic questions on 
how many samples are needed to obtain state estimates. 
Such questions have also been addressed in the 
context of sampled-data control under irregular sampling; see, e.g., 
\cite{Wang2011,Park2011, Rohr2014, Zeng2017,Jungers2018}.


We provide several sufficient conditions on DoS duration and frequency for
the derivation of initial state bounds under DoS attacks.
In the first approach, we analyze the generalized observability matrix
by exploiting a periodic property of the eigenvalues of the system matrix.
Next, we design coders that construct initial state bounds only from consecutive measurements.
Finally, applying the results in \cite{Jungers2018},
we see that if the lengths of DoS periods are bounded, then 
the problem of whether or not a state bound can be constructed is decidable.
All of these approaches provide initial state bounds in finite time.
Consequently,
the proposed encoding schemes achieve Lyapunov stability if the bounds of
DoS duration and frequency are sufficiently small.

The remainder of this paper is organized as follows. 
The networked control system we consider and assumptions on DoS attacks are introduced
in Section II.
In Section III, we propose output encoding schemes
that achieve exponential convergence of the state and its estimate under DoS attacks.
Section IV is devoted to the derivation of initial state bounds in the presence of DoS.
We present a numerical example in Section V.

The results in Section III
partially appeared in our conference paper \cite{Wakaiki2018ACC}.
Here we provide complete proofs not included
in the conference version and make significant structural improvements.
Moreover, the present paper has additional results on the derivation of initial state bounds
and Lyapunov stability.

\subsubsection*{Notation}
The set of non-negative integers is denoted by $\mathbb{Z}_+$.
We denote by $\varrho(P)$ 
the spectral radius of $P \in \mathbb{C}^{n\times n}$.
Let us denote by $A^{*}$ the 
complex conjugate transpose of $A \in \mathbb{C}^{m\times n}$.
For a vector $v \in \mathbb{C}^{ n}$ with $\ell$th element $v_\ell$, 
its maximum norm is  $|v|_{\infty} := \max\{|v_1|,\dots, |v_{ n}|\}$, and
the corresponding induced norm of $A \in \mathbb{C}^{ m\times n}$ 
with $(\ell,j)$th element $A_{\ell j}$
is given by
$\|A\|_{\infty} = \max\{\sum_{j=1}^n|A_{\ell j}| : 1 \leq \ell \leq m\} $.
We denote by ${\rm diag} (\Lambda_1,\dots,\Lambda_n)$ a block diagonal matrix with
diagonal blocks $\Lambda_1,\dots,\Lambda_n$.
For a full column rank matrix $A \in \mathbb{C}^{m\times n}$, 
its left inverse is denoted by $A^{\dagger} = (A^*A)^{-1}A^*$.
A square matrix in $\mathbb{C}^{n \times n}$ is said to be 
{\it Schur stable} if all its eigenvalues lie in
the unit disc.

\section{Networked control system and DoS attack}
\label{sec:closed_loop}
In this section, the networked control system we consider
and 
assumptions on DoS attacks are introduced.
\subsection{Networked control system}
Consider the following discrete-time linear time-invariant 
system:
\begin{subequations}
	\label{eq:plant}
	\begin{align}
	x_{k+1} &= Ax_k + Bu_k \\
	y_k &= Cx_k 
	\end{align}
\end{subequations}
where $x_k \in \mathbb{R}^{n_x}$,
$u_k \in \mathbb{R}^{n_u}$, and
$y_k \in \mathbb{R}^{n_y}$ are 
the state, the input, and the output of the plant, respectively.
The output $y_k$ is encoded and then transmitted through a communication channel subject to DoS.
In contrast, we assume that the input $u_k$ is not affected by any network phenomena,
i.e., $u_k$ goes through the ideal channel.

The decoder sends an acknowledgment to the plant side without delays when
it receives the output data.
If the encoder does not receive the acknowledgment, then
it can detect the DoS attack.
The acknowledgment-based protocol was used in the previous study \cite{Feng2017} on control
without quantization under DoS attacks and also 
has been commonly employed in networked control under 
nonmalicious packet losses; see, e.g., \cite{Imer2006, Tsumura2009}.
Fig.~\ref{fig:closed-loop} illustrates the networked control system we study.

\begin{figure}[tb]
	\centering
	\includegraphics[width = 6cm]{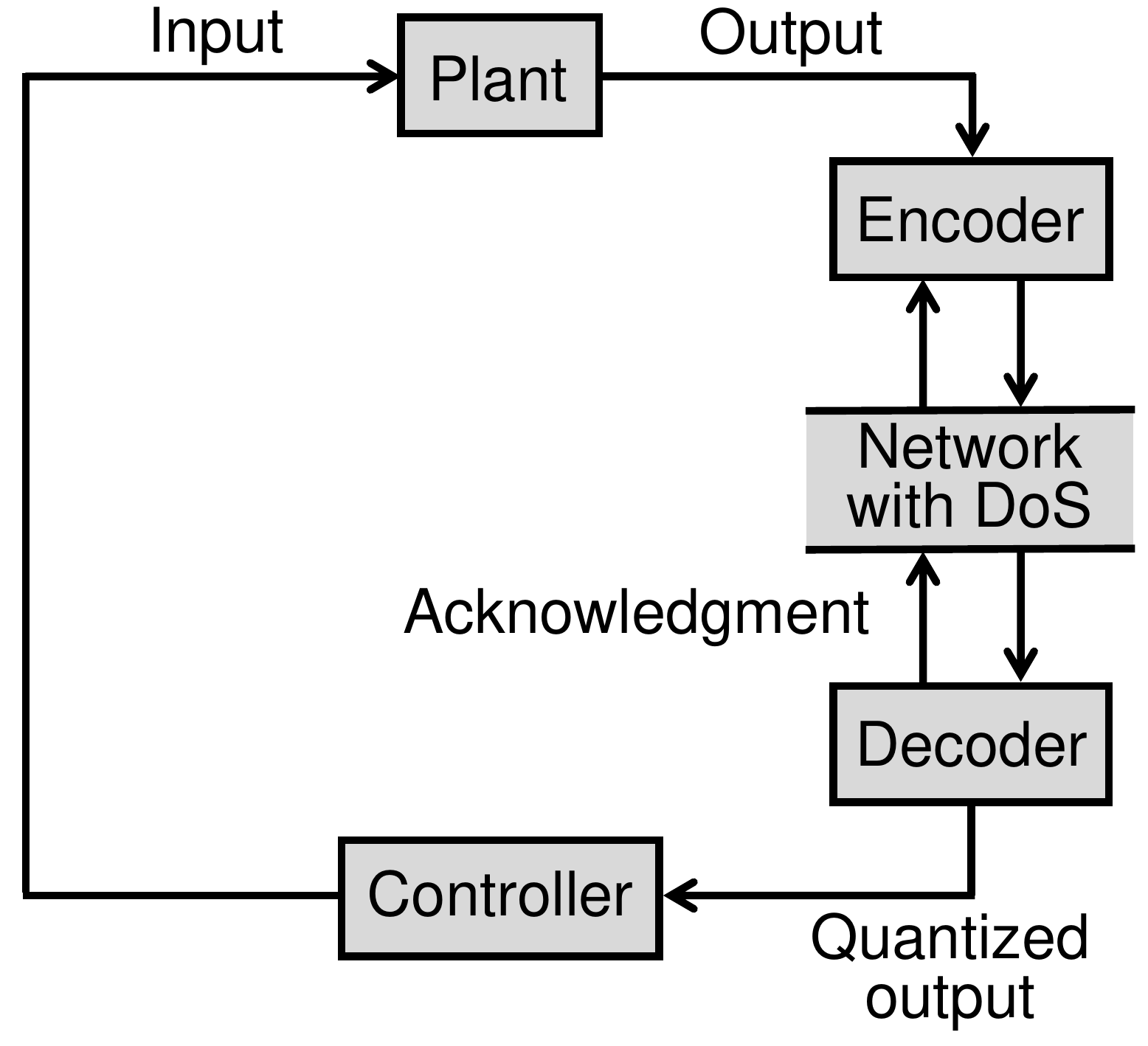}
	\caption{Networked control system under DoS attacks.}
	\label{fig:closed-loop}
	\vspace{-8pt}
\end{figure}

The system matrix $A$ is assumed not to be Schur stable.
This is because if $A$ is Schur stable, then the zero control input 
$u_k = 0$ ($k \in \mathbb{Z}_+$) 
achieves the closed-loop stability for arbitrary DoS attacks, and hence
the stabilization problem we consider would be trivial.

\subsection{DoS attack}
Let us denote by $\Phi_d(k)$ the
number of time-steps when DoS attacks are launched
on the interval $[0,k)$.
As in \cite{	Persis2015,
	Persis2016,
	Feng2017,
	Cetinkaya2017,
	Senejohnny2017,		Kikuchi2017,
	Cetinkaya2018, Lu2018, Feng2018},
we assume that the duration of DoS attacks grows linearly 
with the length of the interval.

\begin{assumption}[Duration of DoS attacks]
	\label{assump:duration}
	There exist $\Pi_d \geq 0$ and $\nu_d \in [0,1]$ such that 
	for every $k \in \mathbb{Z}_+$,
	the DoS duration $\Phi_d(k)$ satisfies
	\begin{equation}
	\label{eq:DOSduration}
	\Phi_d(k) \leq \Pi_d + \nu_d k.
	\end{equation}
	We call $\nu_d$ the {\it DoS duration bound}.
\end{assumption}

The condition \eqref{eq:DOSduration} implies that
at most $\Pi_d + \nu_d k$ packets are affected by DoS attacks
on the interval $[0,k)$.
The DoS duration bound $\nu_d$ is an upper bound of
the limit superior of the DoS duration per time-step.

Next,
let us denote by $\Phi_f(k)$ the number of consecutive DoS attacks on the interval $[0,k)$.
\begin{assumption}[Frequency of DoS attacks]
	\label{assump:freq}
	There exist $\Pi_f \geq 0$ and $\nu_f \in [0,0.5]$ such that 
	for every $k \in \mathbb{Z}_+$,
	the DoS frequency $\Phi_f(k)$ satisfies
	\begin{equation}
	\label{eq:DOSfrequency}
	\Phi_f(k) \leq \Pi_f + \nu_f k.
	\end{equation}
	We call $\nu_f$ the {\it DoS frequency bound}.
\end{assumption}

The DoS frequency bound $\nu_f$ is 
an upper bound of the limit superior of DoS occurrences per time-step.
High-frequency DoS attacks satisfy \eqref{eq:DOSfrequency} with large values of $\nu_f$.

\begin{remark}
	The authors of
	\cite{	Persis2015,
		Persis2016,
		Feng2017,
		Senejohnny2017,Lu2018, Feng2018}
	placed stronger conditions than \eqref{eq:DOSduration} and \eqref{eq:DOSfrequency} such as
	\begin{equation}
	\label{eq:stDOSduration}
	\Phi_d(k,k+\tau) \leq \Pi_d + \nu_d \tau \qquad  \forall k,\tau \in \mathbb{Z_+},
	\end{equation}
	where $\Phi_d(k,k+\tau)$ is the
	number of time-steps when DoS attacks are launched
	on the interval $[k,k+\tau)$. 
	The major reason to place such stronger conditions is that 
	systems with disturbances and noise were considered.
	Although we also consider networked control systems with quantization noise,
	quantization noise decreases under a certain condition on DoS attacks.
	This is the reason why we use the weaker conditions \eqref{eq:DOSduration} and \eqref{eq:DOSfrequency}.
\end{remark}

\section{Exponential convergence under DoS}
\label{sec:exp_conv}
In this section,
we present an encoding and decoding scheme to achieve 
the exponential convergence of the state under the assumption that 
an initial state bound is known.
The proposed schemes are extensions of  the zooming-in method developed in
\cite{Liberzon2003}
to the case under DoS attacks.

We impose the following assumptions throughout this section:
\begin{assumption}[Stabilizability and detectability]
	\label{assump:stabilizability_detectability}
	The pairs $(A,B)$ and $(C,A)$ are stabilizable and detectable, respectively.
	Matrices $K \in \mathbb{R}^{n_u \times n_x}$ and
	$L \in \mathbb{R}^{n_x \times n_y}$ are chosen so that
	$A-BK$ and $A-LC$ are Schur stable.
\end{assumption}

\begin{assumption}[Initial state bound]
	\label{assum:initial_bound}
	A constant $E_{0} > 0$ satisfying 
	$
	|x_0|_{\infty}  \leq E_{0}
	$ is known.
\end{assumption}

An initial bound $E_0$ in Assumption \ref{assum:initial_bound} 
may be given in advance or may be obtained from 
prior measurements via the zooming-out method; see Section IV for the derivation of initial state bounds.

\subsection{Observer-based controller}	
To achieve the exponential convergence of the state,
we use a controller that consists of a Luenberger observer and 
a feedback gain. 
Observer-based controllers update the estimate of
the plant state, by using the output data. However,
when an attack occurs,
the controller cannot receive the output data.
Hence, 
if DoS occurs, then
the controller updates the estimate in the open-loop form.
Define
\begin{equation}
\label{eq:L_def}
L_k := 
\begin{cases}
0 & \text{if DoS occurs at $k$} \\
L & \text{if DoS does not occur at $k$}.
\end{cases}
\end{equation}
The dynamics of the controller is given by
\begin{subequations}
	\label{eq:controller}
	\begin{align}
	\hat x_{k+1} &= 
	A\hat x_k + Bu_k + L_k(q_k - \hat y_k)\\
	u_k &= -K \hat x_k \\
	\hat y_k &= C\hat x_k,
	\end{align}	
\end{subequations}
where $\hat x_k \in \mathbb{R}^{n_x}$,
$\hat y \in \mathbb{R}^{n_y}$, and
$q_k \in \mathbb{R}^{n_y}$
are 
the state estimate, 
the output estimate, and
the quantized value of $y_k$, respectively.
We will provide the details of how to generate 
the quantized output $q_k$ in the next subsection.
We set an initial state estimate $\hat x_0$ to be $\hat x_0 = 0$.

%

\subsection{Basic encoding and decoding scheme}
\label{sec:basic_zooming_in}
Define the error $e_k \in \mathbb{R}^{n_x}$ of the state estimation by
$
e_k := x_k - \hat x_k.
$
Using an invertible matrix $R \in \mathbb{C}^{n_x \times n_x}$, we also define
the transformed error $e_{R,k}\in \mathbb{C}^{n_x}$ by
$
e_{R,k} := Re_k.
$
The invertible matrix $R$ is a design parameter, and
we fix the matrix $R$ arbitrarily in this and next subsections.
Section~\ref{sec:Choice_of_R} includes the discussion on how to choose the matrix $R$.
In particular, we show there that 
if we choose the matrix $R$ that transforms
$A-LC$ into its Jordan canonical form, 
then the assumption on the DoS frequency can be removed.
For this reason, the matrix $R$ is complex-valued.

Let
$E_{R, k} \geq 0$ satisfy
\begin{equation}
\label{eq:ERk_condition}
|e_{R,k}|_{\infty}  \leq E_{R,k}.
\end{equation}
The estimation error of the output is given by
\[
y_k - \hat y_k = Ce_k = CR^{-1} e_{R,k} \qquad \forall k \geq \mathbb{Z}_+.
\]
If the error bound $E_{R,k}$ satisfies \eqref{eq:ERk_condition}, 
then 
\[
|y_k - \hat y_k|_{\infty} \leq \|CR^{-1}\|_{\infty}E_{R,k}.
\]
We partition the hypercube 
\begin{equation}
\label{eq:quantization}
\left\{ y \in \mathbb{R}^{n_y}:| y  - \hat y_k|_{\infty} 
\leq \|CR^{-1}\|_{\infty}  E_{R,k} \right\}
\end{equation}
into $N^{n_y}$ equal boxes.  An index in $\{1,\dots, N^{n_y}\}$ is assigned
to each partitioned box by a certain one-to-one mapping for all $k \in \mathbb{Z}_+$.
The encoder sends to the decoder
the index $q_{k}^{\rm ind}$ of the partitioned box containing 
$y_k$. Then
the decoder generates $q_{k}$ equal to the center of the box
having the index $q_{k}^{\rm ind}$.
If $y_k$ lies on the boundary of several boxes, then
we can choose any one of them.
The quantization error $|y_k - q_k|_{\infty}$ of this encoding scheme satisfies
\begin{equation}
\label{eq:qe_y}
|y_k - q_k|_{\infty}  \leq \frac{\|CR^{-1}\|_{\infty} }{N} E_{R,k}.
\end{equation}

In the next subsection,
we will design a sequence 
$\{E_{R,k}: k \in \mathbb{Z}_+\}$ of error bounds
that achieves \eqref{eq:ERk_condition} for every
$k \in \mathbb{Z}_+$ and exponentially decreases to zero.

\subsection{Main result on exponential convergence}
\label{sec:Duration_Frequency}
Before stating the main result, we first introduce the notion of 
exponential convergence.
\begin{definition}[Exponential convergence]
	The feedback system with the plant \eqref{eq:plant} and the controller \eqref{eq:controller}
	achieves exponential convergence under Assumption \ref{assum:initial_bound} if
	there exist $\Omega \geq 1$ and $\gamma \in (0,1)$, independent of $E_0$, such that 
	\[
	|x_k|_{\infty},~|\hat x_k|_{\infty} \leq \Omega E_0\gamma^k \qquad \forall k \in \mathbb{Z_+}
	\]
	for every initial state $x_0\in \mathbb{R}^{n_x}$ satisfying $|x_0|_{\infty} \leq E_0$.
\end{definition}

Let us introduce an update rule of $\{E_{R,k}:k \in \mathbb{Z}_+\}$ we study here.
Fix an invertible matrix $R \in \mathbb{C}^{n_x \times n_x}$, and
choose $M_0 \geq 1$, $M \geq \|RL\|_{\infty}$, and $\rho \in (0,1)$ satisfying
\begin{subequations}
	\label{eq:normcond_all}
	\begin{align}
	\|R(A-LC)^\ell R^{-1} \|_{\infty} &\leq M_0 \rho^\ell \qquad \forall \ell \geq 0 \label{eq:normcond1}\\
	\|R(A-LC)^{\ell}L\|_{\infty}  &\leq M \rho^\ell \label{eq:normcond2}
	\qquad \forall \ell \geq 0.
	\end{align}
\end{subequations}
Define constants $\theta_a, \theta_0, \theta >0$ by
\begin{subequations}
	\label{eq:theta_def_all}
	\begin{align} 
	\theta_a &:=  \|RAR^{-1}\|_{\infty} \label{eq:thetaa_def} \\
	\theta_0 &:= M_0\rho + \frac{M\|CR^{-1}\|_{\infty}}{N} \label{eq:theta0_def} \\
	\theta &:= \rho + \frac{M\|CR^{-1}\|_{\infty}}{N}. \label{eq:theta_def} 
	\end{align}
\end{subequations}
Using these constants, we set
the error bound $\{E_{R,k}:k \in \mathbb{Z}_+\}$ to be
\begin{align}
&E_{R,k+1} := \begin{cases}
\theta_a  E_{R,k}  & \text{if DoS occurs at $k$} \\
\theta_0 E_{R,k} & \text{else if $k=0$ or DoS occurs at $k-1$} \\
\theta E_{R,k} & \text{otherwise} \\
\end{cases}
\label{eq:E_diff_equation}
\end{align}
for all $k \in \mathbb{Z}_+$.
In terms of the initial value $E_{R,0}$,
we have from Assumption~\ref{assum:initial_bound}
that 
\[
|e_{R,0}|_{\infty}   = |Rx_0|_{\infty}  \leq \|R\|_{\infty}E_{0} =: E_{R,0},
\]
where we used $\hat x_0 = 0$.

The following theorem shows that the encoding scheme with the 
above error bound $\{E_{R,k}:k \in \mathbb{Z}_+\}$
achieves exponential convergence.
\begin{theorem}
	\label{thm:state_conv}
	Suppose that Assumptions \ref{assump:duration}, \ref{assump:freq}, \ref{assump:stabilizability_detectability}, and
	\ref{assum:initial_bound} hold.
	If the number of quantization levels $N$ and
	the DoS duration and frequency bounds $\nu_d$ and $\nu_f$ satisfy
	\begin{subequations}
		\label{eq:with_frequency_Nnu}
		\begin{align}
		N &> \frac{M\|CR^{-1}\|_{\infty}}{1-\rho} \label{eq:N_cond} \\
		\nu_d &< \frac{\log(1/\theta)}{\log(\theta_a/\theta)} 
		- \frac{\log(\theta_0/\theta)}{\log(\theta_a/\theta)} \nu_f,
		\label{eq:DoS_cond}
		\end{align}
	\end{subequations}
	then the feedback system achieves exponential convergence under
	the encoding scheme with the error bound $\{E_{R,k}:k \in \mathbb{Z}_+\}$ constructed by
	the update rule \eqref{eq:E_diff_equation}.
	\vspace{0pt}
\end{theorem}

The proof of this theorem is provided in the next subsection.

\subsection{Proof of Theorem \ref{thm:state_conv}}
We begin by showing that \eqref{eq:ERk_condition}  holds in the absence of DoS attacks.
To this end, we use the technique developed in \cite{Wakaiki2017IFAC}.

The following lemma provides a useful representation of
$\{E_{R,k}:k\in \mathbb{Z}_+\}$ in the case without DoS.
\begin{lemma}
	\label{lem:E_representation}
	For $\ell \in \mathbb{Z}_+$, set
	\begin{equation}
	\label{eq:E_update_without_attack}
	E_{R,k+\ell+1} := 
	\begin{cases}
	\theta_0 E_{R,k} & \text{if $\ell=0$} \\
	\theta E_{R,k+\ell} & \text{otherwise}
	\end{cases}
	\end{equation}
	where $\theta_0$ and 
	$\theta$ are defined by \eqref{eq:theta0_def} and \eqref{eq:theta_def}.
	Then 
	\begin{equation}
	\label{eq:Ek_def_naive}
	E_{R,k+\ell} = M_0 \rho^\ell E_{R,k} + 
	\frac{M \|CR^{-1}\|_{\infty} }{N}  \sum_{j = 0}^{\ell-1} 
	\rho^{\ell-j-1} E_{R,k+j}
	\end{equation}	
	for every $\ell \in \mathbb{N}$.
\end{lemma}
\begin{proof}
	If $\ell = 1$, then \eqref{eq:Ek_def_naive} holds by the definition of $\theta_0$.
	The general case follows by induction.  Define
	$\Delta := M\|CR^{-1}\|_{\infty}/N$.
	If \eqref{eq:Ek_def_naive} holds with 
	$\ell = \ell_0\in \mathbb{N}$, then 
	\begin{align*}
	E_{R,k+\ell_0+1} \!&= \theta E_{R,k+\ell_0} 
	=
	\rho  E_{R,k+\ell_0} + \Delta  E_{R,k+\ell_0} \\
	&=
	M_0 \rho^{\ell_0+1} E_{R,k}+\Delta \!\sum_{j=0}^{\ell_0} \rho^{\ell_0-j}E_{R,k+j}.
	\end{align*}
	Thus, we obtain \eqref{eq:Ek_def_naive} holds with 
	$\ell = \ell_0+1$.
\end{proof}

Using the representation of $E_{R,k+\ell}$ in \eqref{eq:Ek_def_naive},
we show that  \eqref{eq:ERk_condition}  is satisfied in the case without DoS attacks.
\begin{lemma}
	\label{Ek_design_withoutDoS}
	Consider the feedback system in the absence of DoS, that is,
	$L_{k+\ell} = L$ in \eqref{eq:controller} for every $\ell \in \mathbb{Z}_+$.
	Assume that $|e_{R,k}|_{\infty} \leq E_{R,k}$, and 
	set $\{E_{R,k+\ell}:\ell \in \mathbb{Z}_+\}$ as in Lemma \ref{lem:E_representation}.
	Then
	$|e_{R,k+\ell}|_{\infty} \leq E_{R,k+\ell}$ for all $\ell \in \mathbb{Z}_+$.
\end{lemma}
\begin{proof}
	We see from \eqref{eq:plant}
	and \eqref{eq:controller} that
	the state estimation error $e_k$ satisfies
	\begin{equation*}
	e_{k+1} = (A-LC)e_{k} + L(y_{k} - q_{k} ).
	\end{equation*}
	Since $e_{R,k} = Re_k$, it follows that
	\begin{equation}
	\label{eq:error_dynamics}
	e_{R,k+1} = R(A-LC)R^{-1}e_{R,k} + RL(y_{k} - q_{k} ).
	\end{equation}
	Applying induction to \eqref{eq:error_dynamics},
	we obtain
	\begin{align*}
	e_{R,k+\ell} &= R(A-LC)^\ell R^{-1} e_{R,k} + \sum_{j = 0}^{\ell-1} 
	R (A-LC)^{\ell - j - 1} L (y_{k+j} - q_{k+j})
	\end{align*}
	for every $\ell \in \mathbb{N}$.
	It follows from \eqref{eq:qe_y} that
	\begin{align*}
	|e_{R,k+\ell}|_{\infty} &\leq \|R(A-LC)^\ell R^{-1}\|_{\infty}E_{R,k} \\
	& \qquad + \sum_{j = 0}^{\ell-1} 
	\|R(A-LC)^{\ell-j - 1}L\|_{\infty}  \frac{\|CR^{-1}\|_{\infty} }{N} E_{R,k+j}
	\end{align*}
	for every $\ell \in \mathbb{N}$.
	Using the norm condition \eqref{eq:normcond_all},
	we further have
	\begin{equation}
	\label{eq:ek_bound}
	|e_{R,k+\ell}|_{\infty} \leq M_0 \rho^\ell E_{R,k} \!+ \!
	\frac{M \|CR^{-1}\|_{\infty} }{N} 
	\sum_{j = 0}^{\ell-1}\rho^{\ell-j-1} E_{R,k+j}
	\end{equation}
	for every $\ell \in \mathbb{N}$.
	By Lemma \ref{lem:E_representation}, we obtain
	$|e_{R,k+\ell}|_{\infty} \leq E_{R,k+\ell}$ for all $\ell \in \mathbb{Z}_+$.
\end{proof}

Next we investigate the error bound in the presence of DoS attacks.
\begin{lemma}
	\label{lem:Ek_design_DoS}
	Consider the closed-loop system in the presence of DoS, that is,
	$L_{k+\ell} = 0$ in \eqref{eq:controller} for every $\ell \in \mathbb{Z}_+$.
	Assume that $|e_{R,k}|_{\infty} \leq E_{R,k}$, and 
	set 
	\begin{equation}
	\label{eq:E_update_with_attack}
	E_{R,k+\ell+1} := \theta_a E_{R,k+\ell}\qquad \forall \ell \in \mathbb{Z}_+,
	\end{equation}
	where $\theta_a$ is defined as in \eqref{eq:thetaa_def}.
	Then 
	$|e_{R,k+\ell}|_{\infty} \leq E_{R,k+\ell}$ for all $\ell \in \mathbb{Z}_+$.
\end{lemma}
\begin{proof}
	The estimation error $e_k$ satisfies
	$
	e_{k+1} = A e_k,
	$
	and hence 
	\begin{equation}
	\label{eq:error_dynamics_DoS}
	e_{R,k+1} = RAR^{-1} e_{R,k}.
	\end{equation}
	This yields
	\[
	|e_{R,k+1}|_{\infty} \leq
	\|RAR^{-1}\|_{\infty} \cdot |e_{R,k}|_{\infty}
	\leq 
	\|RAR^{-1}\|_{\infty}E_{R,k}.
	\]
	By induction, we obtain
	$|e_{R,k+\ell}|_{\infty} \leq E_{R,k+\ell}$ for every $\ell \in \mathbb{Z}_+$.
\end{proof}

We immediately obtain the following result from
Lemmas \ref{Ek_design_withoutDoS} and \ref{lem:Ek_design_DoS}:
\begin{lemma}
	\label{lem:ERk_bound}
	For the transformed estimation error $e_{R,k}$,
	the error bound $\{E_{R,k}:k \in \mathbb{Z}_+\}$ defined by \eqref{eq:E_diff_equation} satisfies
	\begin{equation}
	\label{eq:eRk<ERk}
	|e_{R,k}|_{\infty} \leq E_{R,k} \qquad \forall k \in \mathbb{Z}_+.
	\end{equation}
	\vspace{0pt}
\end{lemma}

Next, we show that 
the error bound $\{E_{R,k}:k \in \mathbb{Z}_+\}$ in \eqref{eq:E_diff_equation} converges to zero
under the condition \eqref{eq:with_frequency_Nnu}.
\begin{lemma}
	\label{lem:ERk_conv}
	Under the same hypotheses of Theorem \ref{thm:state_conv},
	there exist $\Omega \geq 1$ and $\gamma \in (0,1)$ such that
	\begin{equation}
	\label{eq:ERk_conv}
	E_{R,k} \leq  \Omega E_{R,0} \gamma^k\qquad \forall k \in \mathbb{Z}_+.
	\end{equation}	
	\vspace{0pt}
\end{lemma}
\begin{proof}
	Choose $k_e \in \mathbb{N}$ arbitrarily, and
	assume that  DoS attacks are launched at
	\[
	k=k_m,\dots,k_m+\tau_m-1\qquad \forall m=1,\dots,p
	\]
	on the interval $[0,k_e)$, where
	$k_m \in \mathbb{Z}_+$, $\tau_m \in \mathbb{N}$ for
	every $m=1,\dots,p$ and
	\begin{gather*}
	k_m+\tau_m < k_{m+1}\qquad \forall m=1,\dots,p-1.
	\end{gather*}
	Namely, $k_m$ and $\tau_m$ are the beginning time and the length of 
	the $m$th DoS interval.
	Here $\sum_{m=1}^p \tau_m$ is the total duration of DoS attacks on $[0,k_e)$,
	and  $p$ is the total number of consecutive DoS attacks on $[0,k_e)$.
	Therefore, $\sum_{m=1}^p \tau_m = \Phi_d(k_e)$ and $p = \Phi_f(k_e)$.
	
	In what follows, we assume that $k_1 >0 $ and $k_p +\tau_p < k_e$ for simplicity. 
	In the case where $k_1 = 0$ or $k_p +\tau_p  = k_e$,
	one can prove the convergence of the error bound \eqref{eq:ERk_conv}
	in a similar way.
	
	Define 
	\begin{align*}
	r_1 &:= k_1-1, \quad r_{p+1} := k_e-k_p-\tau_p-1 \\
	r_m &:= k_m-k_{m-1}-\tau_{m-1}-1  \qquad \forall m=2,\dots,p.
	\end{align*}
	Then $r_m \geq 0$ for every $m=1,\dots,p+1$.
	Since DoS attacks are not launched on the interval $[0,k_1)$,
	it follows that
	\begin{align*}
	E_{R,k_1} = \theta^{r_1}\theta_0 E_{R,0}.
	\end{align*}
	On the other hand, DoS occurs on the interval $[k_1,\dots,k_1+\tau_1)$,
	and hence
	\begin{align*}
	E_{R,k_1+\tau_1} = \theta_a^{\tau_1}E_{R,k_1} 
	=
	\theta_a^{\tau_1}\theta^{r_1}\theta_0 E_{R,0}.
	\end{align*}
	Continuing in this way,
	we see that the error bound $E_{R,k_e}$ at the time $k=k_e$ satisfies
	\begin{align}
	E_{R,k_e} &= \theta^{r_{p+1}} \theta_0 E_{R,k_p+\tau_p} \notag \\
	&= \theta^{\sum_{m=1}^{p+1} r_m} \cdot \theta_0^{p+1} \cdot
	\theta_a^{\sum_{m=1}^{p} \tau_m}E_{R,0}. \label{eq:ET}
	\end{align}
	By definition,
	\begin{equation}
	\label{eq:ri_sum}
	\sum_{m=1}^{p+1} r_m = k_e - (p+1) - \sum_{m=1}^{p} \tau_m.
	\end{equation}
	Moreover, it follows from Assumptions \ref{assump:duration} and 
	\ref{assump:freq} that
	\begin{align}
	\label{eq:dp_bound}
	\sum_{m=1}^{p} \tau_m \leq \Pi_d + \nu_d k_e,\quad
	p  \leq \Pi_f + \nu_f k_e. 
	\end{align}
	Substituting \eqref{eq:ri_sum} and \eqref{eq:dp_bound} into \eqref{eq:ET}, we obtain
	\begin{align*}
	E_{R,k_e} 
	&=
	\theta^{k_e} \cdot
	\left(
	\frac{\theta_0}{\theta}
	\right)^{p+1} \cdot
	\left(
	\frac{\theta_a}{\theta}
	\right)^{\sum_{m=1}^{p} \tau_m}E_{R,0}  \\
	&\leq
	\frac{\theta_0^{\Pi_f+1} \cdot \theta_a^{\Pi_d}}{ \theta^{\Pi_f+\Pi_d+1}}
	\left(
	\theta \cdot 
	\left(
	\frac{ \theta_0}{\theta}
	\right)^{\nu_f} \cdot
	\left(
	\frac{\theta_a}{\theta}
	\right)^{\nu_d}
	\right)^{k_e} E_{R,0}.
	\end{align*}
	Since the inequality \eqref{eq:DoS_cond} is equivalent to
	\[
	\theta \cdot 
	\left(
	\frac{\theta_0}{\theta}
	\right)^{\nu_f} \cdot
	\left(
	\frac{\theta_a}{\theta}
	\right)^{\nu_d} < 1,
	\]
	the exponential convergence of the error bound
	\eqref{eq:ERk_conv} is established.
\end{proof}

We are now in a position to prove Theorem \ref{thm:state_conv}. 

\begin{proof}[Proof of Theorem \ref{thm:state_conv}]
The state $x_k$ satisfies
\begin{align}
x_{k+1} 
&= (A-BK)^{k+1} x_0 + 
\sum_{\ell=0}^k
(A-BK)^{k-\ell} BK R^{-1}e_{R,\ell} \notag.
\end{align}
Therefore,
\begin{align}
&|x_{k+1}|_{\infty}
\leq \|(A-BK)^{k+1}\|_{\infty}\cdot  |x_0|_{\infty} + 
\sum_{\ell=0}^k
\|(A-BK)^{k-\ell}\|_{\infty}\cdot  \| BKR^{-1}\|_{\infty} \cdot |e_{R,\ell}|_{\infty}.
\label{eq:x_norm_bound}
\end{align}
By Lemmas \ref{lem:ERk_bound} and  \ref{lem:ERk_conv}, there exist $\Omega \geq 1$ and $\gamma \in (0,1)$
such that 
\begin{equation}
\label{eq:e_norm_bound}
|e_{R,\ell} |_{\infty} \leq \Omega E_{R,0} \gamma^\ell \qquad \forall \ell \in \mathbb{Z}_+.
\end{equation}
Moreover, since $A-BK$ is Schur stable by Assumption \ref{assump:stabilizability_detectability}, 
there exist
$\Omega_K \geq 1$ and $\tilde \gamma \in [\gamma,1)$ such that 
\begin{equation}
\label{eq:ABK_bound}
\|(A-BK)^{\ell}\|_{\infty} \leq \Omega_K \tilde \gamma^{\ell}
\qquad \forall \ell \in \mathbb{Z}_+.
\end{equation}
Substituting \eqref{eq:e_norm_bound} and \eqref{eq:ABK_bound}
into \eqref{eq:x_norm_bound}, we obtain
\begin{align}
|x_{k+1}|_{\infty} &\leq
\Omega_K\tilde \gamma^{k+1} |x_0|_{\infty} +
\Omega \Omega_K E_{R,0} \|BKR^{-1}\|_{\infty} (k+1) \tilde \gamma^{k}.
\label{eq:x_norm_bound_q}
\end{align}
For every $\varepsilon > 0$,
there exists a constant $\alpha \geq 1$ such that 
$k \tilde \gamma^k \leq  \alpha (\tilde \gamma+ \varepsilon)^k$
for all $k \in \mathbb{Z}_+$. Thus 
\eqref{eq:x_norm_bound_q} leads to the 
exponential convergence of the state.
Additionally, since $\hat x_k = x_k - R^{-1}e_{R,k}$, it follows that
$\hat x_k$ also exponentially converges to zero.
This completes the proof. 
\end{proof}


\begin{remark}
	In the previous studies \cite{	Persis2015,
		Persis2016,
		Feng2017,
		Senejohnny2017, Lu2018, Feng2018},
	the assumption on the frequency of DoS attacks is used in a different way.
	The above studies consider continuous-time attacks, and hence
	the frequency at which DoS attacks are launched must be 
	smaller than the sampling rate. 
	Therefore, in the discrete-time case \cite{Cetinkaya2017},
	the assumption on the DoS frequency bound 
	is not used.
	However, the output encoding scheme in Theorem \ref{thm:state_conv}
	increases the error bound, $E_{R,k+1} = \theta_0E_{R,k}$, at the first time-step after DoS attacks occur.
	For this reason, we here employ the frequency assumption to obtain a less conservative sufficient condition.
\end{remark}

\subsection{Choice of invertible matrix R}
\label{sec:Choice_of_R}
In this subsection, we provide a guideline for choosing the invertible matrix $R \in \mathbb{C}^{n_x\times n_x}$.
We show that if the matrix $R$ is chosen appropriately, then
the encoding scheme in Theorem~\ref{thm:state_conv} 
does not need the assumption of the DoS frequency.
To this end,
we first provide a basic fact of the maximum norm.
\begin{proposition}
	\label{prop:Rc_existence}
	For a matrix $\Xi \in  \mathbb{C}^{n \times n}$ and
	a scalar $\varepsilon > 0$, take an 
	invertible matrix $R \in \mathbb{C}^{n \times n}$ satisfying
	\begin{align}
	\label{eq:cA_Jordan}
	\frac{1}{\varepsilon} R\Xi R^{-1} = J,
	\end{align}	
	where $J$ is the Jordan canonical form of $\Xi/\varepsilon$.
	Then the matrix $R$ satisfies 
	\begin{equation}
	\label{eq:max_norm_prop}
	\|R \Xi R^{-1}\|_{\infty} \leq  \varrho (\Xi) + \varepsilon.
	\end{equation}
	\vspace{0pt}
\end{proposition}

\begin{proof}
	Let us denote the eigenvalues of $\Xi$ by 
	$\lambda_1,\dots,\lambda_n$ (including multiplicity).
	Then the diagonal part of the Jordan canonical form 
	$J$ consists of $\lambda_1/\varepsilon,\dots,\lambda_n/\varepsilon$.
	Therefore, 
	\begin{equation}
	\label{eq:Jc_c}
	\left\|\varepsilon  J  \right\|_{\infty} \leq  \max_{j=1,\dots,n} |\lambda_j|  +\varepsilon 
	= \varrho(\Xi) + \varepsilon.
	\end{equation}
	By \eqref{eq:cA_Jordan} and \eqref{eq:Jc_c}, we obtain the desired conclusion.
\end{proof}

In particular, if the matrix $\Xi$ is Schur stable in Proposition~\ref{prop:Rc_existence}, 
then we obtain the following result by
choosing a sufficiently small  $\varepsilon >0$.
\begin{corollary}
	\label{coro:stable_case}
	For every Schur stable matrix $\Xi \in  \mathbb{C}^{n \times n}$, there exists 
	an invertible matrix $R\in \mathbb{C}^{n \times n}$ such that
	$\|R\Xi R^{-1}\|_{\infty}  < 1$.
\end{corollary}

Let
an invertible matrix $R \in \mathbb{C}^{n_x \times n_x}$ satisfy
\begin{align}
\label{eq:R_ALC_cond}
\|R(A-LC)R^{-1}\|_{\infty} < 1.
\end{align}
Corollary \ref{coro:stable_case} shows 
that such a matrix $R$ always exists under Assumption \ref{assump:stabilizability_detectability}.
We set
the error bound $\{ E_{R,k}:k \in \mathbb{Z}_+\}$ to be 
\begin{align}
E_{R,k+1} := 
\begin{cases}
\vartheta_a   E_{R,k}  & \text{if DoS occurs at $k$} \\
\vartheta  E_{R,k}  & \text{otherwise}, 
\end{cases}
\label{eq:E_diff_equation_coro}
\end{align}
where
\begin{subequations}
	\begin{align}
	\vartheta_a & :=  \|RAR^{-1}\|_{\infty} 	\label{eq:thetaa_def_coro} \\
	\vartheta & :=  \|R(A-LC)R^{-1}\|_{\infty}   +  \frac{\|RL\|_{\infty}  \cdot  \|CR^{-1}\|_{\infty}}{N}. \label{eq:theta_def_coro}
	\end{align}
\end{subequations}

The following result, which is a corollary of Theorem \ref{thm:state_conv},
shows that the encoding scheme with the error bound $\{ E_{R,k}: k \in \mathbb{Z}_+\}$
updated by \eqref{eq:E_diff_equation_coro}
achieves exponential convergence without any DoS frequency assumptions.
\begin{corollary}
	\label{coro:simle_case}
	Suppose that Assumptions \ref{assump:duration}, \ref{assump:stabilizability_detectability}, and
	\ref{assum:initial_bound} hold.
	Assume that 
	an invertible matrix $R \in \mathbb{C}^{n_x \times n_x}$ satisfies
	\eqref{eq:R_ALC_cond}.
	If the number of quantization levels $N$
	and the DoS duration bound $\nu_d$ satisfy 
	\begin{subequations}
		\label{eq:without_frequency_Nnu}
		\begin{align}
		N &> \frac{\|RL\|_{\infty} \cdot \|CR^{-1}\|_{\infty}}{1-\|R(A-LC)R^{-1}\|_{\infty} } \label{eq:N_cond_coro} \\
		\label{eq:DoS_cond_coro}
		\nu_d &< \frac{\log(1/\vartheta)}{\log(\vartheta_a/\vartheta)},
		\end{align}
	\end{subequations}
	then the feedback system achieves exponential convergence under
	the encoding scheme with the error bound $\{E_{R,k}:k \in \mathbb{Z}_+\}$ constructed by
	the update rule \eqref{eq:E_diff_equation_coro}.
	\vspace{0pt}
\end{corollary}
\begin{proof}
	We can set the constants $\rho, M_0,M$ in \eqref{eq:normcond_all} to be 
	\begin{align*}
	\rho = \|R(A-LC)R^{-1}\|_{\infty},~~
	M_0 = 1,~~
	M = \|RL\|_{\infty}.
	\end{align*}
	Then $\theta_0$ and $\theta$ defined in \eqref{eq:theta0_def}, \eqref{eq:theta_def}
	are equal to $\vartheta$ in \eqref{eq:theta_def_coro}.
	By definition, $\theta_a = \vartheta_a$.
	Since $\log(\theta_0/\theta) = 0$, 
	the conditions \eqref{eq:without_frequency_Nnu} on $N$ and $\nu_d$
	are the same as the conditions \eqref{eq:with_frequency_Nnu}.
	Thus, the desired result follows from Theorem \ref{thm:state_conv}.
\end{proof}

\begin{remark}
	If the pair $(C,A)$ is observable, then there exists a deadbeat gain 
	$L \in \mathbb{R}^{n_x \times n_y}$ such that $\varrho(A-LC) = 0$.
	Proposition \ref{prop:Rc_existence} shows that
	for every $\varepsilon > 0$, 
	there exists an invertible matrix $R \in \mathbb{C}^{n_x \times n_x}$ such that
	\[
	\|R(A-LC)R^{-1}\|_{\infty} < \varepsilon.
	\]
	In this case,
	since $\vartheta$ in \eqref{eq:theta_def_coro} satisfies
	$\lim_{N \to \infty} \vartheta \leq \varepsilon$,
	Corollary \ref{coro:simle_case} shows that 
	for every DoS duration bound $\nu_d \in [0,1)$,
	there exists $N \in \mathbb{N}$ 
	such that exponential convergence
	is achieved, which is consistent with Theorem 2 of \cite{Feng2017}.
\end{remark}

\begin{remark}
	In Corollary \ref{coro:simle_case}, 
	we choose the matrix $R$ so that the 
	growth rate $\theta_0$ in \eqref{eq:theta0_def} of $E_{R,k}$ 
	is less than one.
	Another choice of the matrix $R$ is to reduce 
	the other growth rate $\theta_a$
	in \eqref{eq:thetaa_def}. In fact,
	Proposition \ref{prop:Rc_existence} shows that 
	for every $\varepsilon > 0$, 
	we can obtain an invertible matrix $R \in \mathbb{C}^{n_x \times n_x}$
	satisfying
	\[
	\varrho (A) \leq \theta_a = \|RAR^{-1}\|_{\infty} \leq \varrho (A) + \varepsilon.
	\]
	As seen in the numerical example of Section V,
	if the DoS frequency bound $\nu_f$ is sufficiently small, then
	the latter choice  can lead to the update rule \eqref{eq:E_diff_equation}
	that allows longer DoS duration than the update rule \eqref{eq:E_diff_equation_coro}.
\end{remark}

\begin{remark}
	The encoding scheme of Theorem \ref{thm:state_conv} has
	a freedom in the choice of the matrix $R$, whereas 
	the matrix $R$ in  Corollary \ref{coro:simle_case} needs to satisfy \eqref{eq:R_ALC_cond}
	but allows stability analysis without any assumption on the DoS frequency.
	In general, we cannot say which encoding scheme is better with respect to the quantization level $N$
	and the DoS duration bound $\nu_d$. Moreover, it is difficult to design the matrix $R$ 
	satisfying given conditions on the quantization
	level $N$ and the DoS duration bound $\nu_d$ without employing metaheuristics  such as
	genetic algorithms.
	We leave this issue for future investigation.
\end{remark}

\subsection{Encoding scheme with center at origin}
In Sections \ref{sec:basic_zooming_in}--\ref{sec:Choice_of_R}, we have considered 
the encoding scheme that 
uses the output estimate as the quantization center.
Here we propose encoding schemes with center at
the origin.
In such an encoding scheme, the encoder does not need
to compute the output estimate. Therefore,
we can encode the output with less computational resources.

Define 
\begin{align*}
z_k &:= 
\begin{bmatrix}
x_{k} \\ e_k
\end{bmatrix} \qquad 
C_{\text{cl}} := 
\begin{bmatrix}
C & 0
\end{bmatrix}.
\end{align*}
Using an invertible matrix $R_{\text{cl}} \in \mathbb{C}^{2n_x \times 2n_x}$, 
we define the transformed closed-loop state $z_{R,k} \in \mathbb{C}^{2n_x}$
by $z_{R,k} := R_{\text{cl}} z_k$.
Let $E_{R,k}^z \geq 0$ satisfy
\begin{equation}
\label{eq:z_Rk_bound}
|z_{R,k}|_{\infty} \leq E_{R,k}^z.
\end{equation}
Then 
$
|y_k|_{\infty} = |C_{\text{cl}}R_{\text{cl}}^{-1} z_{R,k}|_{\infty} \leq \|C_{\text{cl}}R_{\text{cl}}^{-1}\|_{\infty} E_{R,k}^z.
$

The only difference from the encoding and decoding scheme in Section III-B
is that we here employ the hypercube with center at the origin
\begin{equation*}
\left\{ y \in \mathbb{R}^{n_y}:| y |_{\infty} 
\leq \|C_{\text{cl}}R_{\text{cl}}^{-1}\|_{\infty}  E_{R,k}^z \right\},
\end{equation*}
and partition it
into $N^{n_y}$ equal boxes, instead of
the hypercube with center at the output estimate \eqref{eq:quantization}.
The quantization error $|y_k - q_k|_{\infty}$ of this encoding scheme satisfies
\begin{equation*}
|y_k - q_k|_{\infty}  \leq \frac{\|C_{\text{cl}}R_{\text{cl}}^{-1}\|_{\infty} }{N} E_{R,k}^z.
\end{equation*}

To
achieve the exponential convergence of the closed-loop state $z_k$,
we aim at designing a sequence $\{E_{R,k}^z:k \in \mathbb{Z}_+\}$ of state bounds
that satisfies \eqref{eq:z_Rk_bound} for every $k \in \mathbb{Z}_+$ and
exponentially decreases to zero.  We start with the dynamics of the 
transformed closed-loop state $z_{R,k}$.
Define
\begin{align*}
A_{\text{cl}} &:= 
\begin{bmatrix}
A-BK & BK \\
0 & A-L C
\end{bmatrix},\quad 
A_{\text{op}} := 
\begin{bmatrix}
A-BK & BK \\
0 & A
\end{bmatrix},\quad
L_{\text{cl}} :=
\begin{bmatrix}
0 \\
L
\end{bmatrix}.
\end{align*}
If DoS does not occur at time $k$, then
\[
z_{R,k+1} = R_{\text{cl}}A_{\text{cl}} R_{\text{cl}}^{-1} z_{R,k} +
R_{\text{cl}} L_{\text{cl}} (y_k - q_k);
\]
otherwise
\[
z_{R,k+1} = R_{\text{cl}}A_{\text{op}} R_{\text{cl}}^{-1} z_{R,k}.
\]
The dynamics of the closed-loop state $z_{R,k}$ 
has the same structure as that of the error $e_{R,k}$ in 
\eqref{eq:error_dynamics} and \eqref{eq:error_dynamics_DoS}.
Therefore, we apply the discussion in Sections 
\ref{sec:Duration_Frequency}--\ref{sec:Choice_of_R} to
the sequence $\{E_{R,k}^z:k\in \mathbb{Z}_+\}$ of state bounds
with minor modifications.

First, we introduce the counterpart of the encoding scheme in Theorem \ref{thm:state_conv}.
Choose $M_0^z \geq 1$, $M^z \geq \|R_{\text{cl}}L_{\text{cl}}\|_{\infty}$, and $\rho_{\text{cl}} \in (0,1)$ satisfying
\begin{subequations}
	\label{eq:normcond_all_origin}
	\begin{align}
	\|R_{\text{cl}}A_{\text{cl}}^\ell R_{\text{cl}}^{-1} \|_{\infty} &\leq M_0^z \rho_{\text{cl}}^\ell \qquad \forall \ell \geq 0\\
	\|R_{\text{cl}}A_{\text{cl}}^{\ell}L_{\text{cl}}\|_{\infty}  &\leq M^z \rho_{\text{cl}}^\ell 
	\qquad \forall \ell \geq 0.
	\end{align}
\end{subequations}
Define constants $\phi_a, \phi_0, \phi > 0$ by
\begin{align*} 
\phi_a &:=  \|R_{\text{cl}}A_{\text{op}}R_{\text{cl}}^{-1}\|_{\infty} \\
\phi_0 &:= M_0^z\rho_{\text{cl}} + \frac{M^z\|C_{\text{cl}}R_{\text{cl}}^{-1}\|_{\infty}}{N} \\
\phi &:= \rho_{\text{cl}} + \frac{M^z\|C_{\text{cl}}R_{\text{cl}}^{-1}\|_{\infty}}{N}. 
\end{align*}
Using these constants, we set
the sequence $\{E_{R,k}^z:k \in \mathbb{Z}_+\}$ of state bounds to be
\begin{align}
E_{R,k+1}^z := 
\begin{cases}
\phi_a  E_{R,k}^z  & \text{if DoS occurs at $k$} \\
\phi_0 E_{R,k}^z & \text{else if $k=0$ or DoS occurs at $k-1$} \\
\phi E_{R,k}^z & \text{otherwise}\\
\end{cases}
\label{eq:E_diff_equation_origin}
\end{align}
for all $k \in \mathbb{Z}_+$.
Since $\hat x_0 = 0$, it follows that $|z_{0}|_{\infty} \leq E_0$  under Assumption~\ref{assum:initial_bound}.
Therefore,
\[
|z_{R,0}|_{\infty} = |R_{\text{cl}}z_0|_{\infty} \leq \|R_{\text{cl}}\|_{\infty} E_0 
=: E_{R,0}^z.
\]

\begin{theorem}
	\label{thm:state_conv_origin}
	Suppose that  
	Assumptions \ref{assump:duration}, \ref{assump:freq}, \ref{assump:stabilizability_detectability}, and
	\ref{assum:initial_bound} hold.
	If the number of quantization levels $N$ and
	the DoS duration and frequency bounds $\nu_d$ and $\nu_f$ satisfy
	\begin{subequations}
		\label{eq:with_frequency_Nnu_origin}
		\begin{align}
		N &> \frac{M^z\|C_{\text{cl}}R_{\text{cl}}^{-1}\|_{\infty}}{1-\rho_{\text{cl}}} \label{eq:N_cond_origin} \\
		\nu_d &< \frac{\log(1/\phi)}{\log(\phi_a/\phi)} 
		- \frac{\log(\phi_0/\phi)}{\log(\phi_a/\phi)} \nu_f,
		\label{eq:DoS_cond_origin}
		\end{align}
	\end{subequations}
	then the feedback system achieves exponential convergence under
	the encoding scheme with the state bound $\{E_{R,k}^z:k \in \mathbb{Z}_+\}$ constructed by
	the update rule \eqref{eq:E_diff_equation_origin}.
	\vspace{0pt}
\end{theorem}
\begin{proof}
	As in Lemmas \ref{lem:ERk_bound} and \ref{lem:ERk_conv},
	one can show that \eqref{eq:z_Rk_bound} holds for every $k \in \mathbb{Z}_+$
	and that
	there exist $\Omega \geq 1$ and $\gamma \in (0,1)$ such that
	\begin{equation*}
	E_{R,k}^z \leq  \Omega E_{R,0}^z \gamma^k\qquad \forall k \in \mathbb{Z}_+.
	\end{equation*}	
	By the definition of $z_{R,k}$, these facts yield
	the exponential convergence of the closed-loop system.
\end{proof}

The decay rate $\rho$ in \eqref{eq:normcond_all} depends only on $A-LC$ and
satisfies $\rho \geq \varrho(A-LC)$. In contrast,
the counterpart $\rho_{\rm cl}$ in \eqref{eq:normcond_all_origin}
depends on $A-BK$ as well as $A-LC$, and 
$\rho_{\rm cl} \geq \max\{\varrho(A-BK),\varrho(A-LC) \}$ holds.
Therefore, when 
$\varrho(A-LC)$ is small but
$\varrho(A-BK)$ is large, 
the encoding scheme with center at the origin decreases
quantization errors slowly.

Next we present the counterpart of the encoding scheme in Corollary \ref{coro:simle_case}.
We set
the state bound $\{ E_{R,k}^z:k \in \mathbb{Z}_+\}$ to be 
\begin{align}
E_{R,k+1}^z := 
\begin{cases}
\varphi_a   E_{R,k}^z  & \text{if DoS occurs at $k$} \\
\varphi  E_{R,k}^z  & \text{otherwise}, 
\end{cases}
\label{eq:E_diff_equation_coro_origin}
\end{align}
where
\begin{align*}
\varphi_a &:= \|R_{\text{cl}}A_{\text{op}}R_{\text{cl}}^{-1}\|_{\infty} \\
\varphi &:= \|R_{\text{cl}}A_{\text{cl}}R_{\text{cl}}^{-1}\|_{\infty}  
+ \frac{\|R_{\text{cl}}L_{\text{cl}}\|_{\infty} \cdot \|C_{\text{cl}}R_{\text{cl}}^{-1}\|_{\infty}}{N}.
\end{align*}

We obtain the following corollary of Theorem \ref{thm:state_conv_origin} from the same argument as in
Corollary \ref{coro:simle_case}.
\begin{corollary}
	\label{coro:simle_case_origin}
	Suppose that Assumptions \ref{assump:duration}, \ref{assump:stabilizability_detectability}, and
	\ref{assum:initial_bound} hold.
	Assume that 
	an invertible matrix $R \in \mathbb{C}^{2n_x \times 2n_x}$ satisfies
	\[
	\|R_{\text{cl}}A_{\text{cl}}R_{\text{cl}}^{-1}\|_{\infty}  < 1.
	\]
	If the number of quantization levels $N$
	and the DoS duration bound $\nu_d$ satisfy 
	\begin{subequations}
		\label{eq:without_frequency_Nnu_origin}
		\begin{align}
		N &> \frac{\|R_{\text{cl}}L_{\text{cl}}\|_{\infty} 
			\cdot \|C_{\text{cl}}R_{\text{cl}}^{-1}\|_{\infty}}{1-\|R_{\text{cl}}A_{\text{cl}}R_{\text{cl}}^{-1}\|_{\infty} } \label{eq:N_cond_coro_origin} \\
		\label{eq:DoS_cond_coro_origin}
		\nu_d &< \frac{\log(1/\varphi)}{\log(\varphi_a/\varphi)},
		\end{align}
	\end{subequations}
	then the feedback system achieves exponential convergence under
	the encoding scheme with the state  bound $\{E_{R,k}^z:k \in \mathbb{Z}_+\}$ constructed by
	the update rule \eqref{eq:E_diff_equation_coro_origin}.
\end{corollary}

\section{Derivation of initial state bounds under DoS}
In this section, we present an encoding and decoding scheme
to obtain an initial state bound $E_0$ under DoS attacks.
To this end, we extend the zooming-out 
method proposed in \cite{Liberzon2003} to
the case under DoS attacks.
Due to the attacks and the missing output data,
it can be difficult to obtain correct state estimates on the
decoder side, which is a basic question related to observability.

We place the following assumptions in this section:
\begin{assumption}[Observability]
	\label{assump:observability}
	The pair $(C,A)$ is observable.
\end{assumption}
\begin{assumption}[Odd quantization level]
	\label{assump:odd}
	The quantization level $N$ is an odd number.
\end{assumption}

The derivation of state bounds requires observability rather than detectability.
If the quantization number $N$ is odd, then
the quantized value $q_k$ is zero for a sufficiently small output $y_k$.
We use this property for Lyapunov stability in Section~\ref{sec:Lyap_stability}.

\subsection{Basic encoding and decoding scheme}
\label{sec:basic_zooming_out}
We set the control input $u_k$ to be $u_k = 0$
until we get a state bound. 
For a given increasing sequence $\{E_k^y \geq 0:k \in \mathbb{Z}_+\}$,
define the binary function $Q_k:\mathbb{R}^{n_y} \to \{0,1\}$ by
\[
Q_k(y) := 
\begin{cases}
0 & \text{if $|y|_{\infty} \leq E_k^y$} \\
1 & \text{otherwise}.
\end{cases}
\]

For $s_0,s_1,\dots,s_{\chi}  \in \mathbb{Z}_+$ satisfying
$s_{0} < s_1 < \cdots < s_{\chi}$,
we define the generalized observability matrix $O\left( \{s_m\}_{m=0}^{\chi} \right)$ by
\begin{align}
\label{eq:Obs_matrix_DoS}
O\left( \{s_m\}_{m=0}^{\chi} \right) := 
\begin{bmatrix}
C \\ CA^{s_1-s_0} \\ \vdots \\ CA^{s_{\chi} - s_0}
\end{bmatrix}.
\end{align}

Assume that there exists $\{s_m\}_{m=0}^{\chi} \subset \mathbb{Z}_+$ with
$s_{0} < s_1 < \cdots < s_{\chi}$ such that
DoS attacks do not occur at times $ k =s_0,\dots,s_{\chi}$ and
the following two conditions hold:
\begin{description}
	\item[(C1)] $Q_{s_{m}}(y_{s_{m}})  = 0$ holds for every $m=0,\dots,\chi$;
	\item[(C2)] The matrix $O\left( \{s_m\}_{m=0}^{\chi} \right)$ 
	is full column rank.
\end{description}

Under the conditions (C1) and (C2),
we can obtain a state bound at $k= s_{\chi}+1$ as follows.
By the condition (C1),
the decoder on the controller side knows at time $k = s_{\chi}$
that 
\[
|y_{s_m}|_{\infty} \leq E_{s_m}^y \qquad \forall m = 0,\dots,\chi
\]
and hence
\[
\left|
\begin{bmatrix}
y_{s_0} \\
\vdots \\
y_{s_{\chi}}
\end{bmatrix}
\right|_{\infty}
\leq E_{s_{\chi}}^y.
\]
By the condition (C2),
\begin{equation}
\label{eq:x_0_reconstruction_time_varying}
x_{s_0} = O\left( \{s_m\}_{m=0}^{\chi} \right)^{\dagger}
\begin{bmatrix}
y_{s_0} \\ y_{s_1} \\ \vdots \\ y_{s_{\chi}}
\end{bmatrix}.
\end{equation}
Since $x_{s_{\chi}+1} = A^{s_{\chi}-s_0+1} x_{s_0}$, it follows that 
if the coders set 
a state bound $E_{s_{\chi}+1}$ at time $k=s_{\chi}+1$
to be
\begin{equation}
\label{eq:Ek1_eta_bound}
E_{s_{\chi} + 1} := 
\left\|A^{s_{\chi} - s_0+1}O\left( \{s_m\}_{m=0}^{\chi} \right)^{\dagger} \right\|_{\infty}  E_{s_{\chi}}^y,
\end{equation}
then $|x_{s_{\chi}+1}|_{\infty} \leq E_{s_{\chi+1}}$.

Next, we design a sequence $\{E_k^y:k \in \mathbb{Z}_+\}$ that satisfies
$Q_k(y_k) = 0$ for every sufficiently large $k$.
Fix a constant $\kappa >0$ and an initial value $E_0^x > 0$, and
define 
a sequence $\{E_k^x:k \in \mathbb{Z}_+\}$ by
\begin{equation}
\label{eq:E_zoomout_update}
E_{k+1}^x := (1+\kappa) \|A\|_{\infty} E_k^x.
\end{equation}
Since the growth rate of $E_k^x$ is larger than that of $|x_k|_{\infty}$, 
there exists $T \in \mathbb{Z}_+$ such that
\begin{equation}
\label{eq:y_bound_zo}
|y_k|_{\infty} = |Cx_k|_{\infty} \leq \|C\|_{\infty}E_k^x\qquad \forall k\geq T.
\end{equation}
If we set $E_k^y := \|C\|_{\infty}E_k^x$, then
\eqref{eq:y_bound_zo} yields
$Q_k(y_k) = 0$ for all $k \geq T$.
Thus, 
the condition (C1) is always satisfied if $s_0 \geq T$. 

In what follows,
we set the time origin to be $T \in \mathbb{Z}_+$ satisfying \eqref{eq:y_bound_zo}
for simplicity of notation.
Note that 
we can use the DoS conditions 
\eqref{eq:DOSduration} and \eqref{eq:DOSfrequency}  even after shifting the time origin, by
changing the constants $\Pi_d$ and $\Pi_f$ there to
$\Pi_d+\nu_dT$ and $\Pi_f+\nu_fT$, respectively.
\begin{assumption}[Capturing output from initial time]
	\label{assump:initial_time}
	For every $k \in \mathbb{Z}_+$, $Q_k(y_k) = 0$. 
\end{assumption}

In the case without DoS attacks  \cite{Liberzon2003}, 
if $(C,A)$ is observable, then 
we can obtain a state bound at time $k=\eta$, where
$\eta$ is the observability index, because
\begin{equation}
\label{eq:O_def}
O := 
\begin{bmatrix}
C \\
CA \\
\vdots \\
CA^{\eta-1}
\end{bmatrix}
\end{equation}
is full column rank.
However, the following example shows that
for every DoS duration bound $\nu_d \in (0,1]$, there exists
an observable system $(C,A)$ and a corresponding attack strategy such that 
the condition (C2) does not hold.
\begin{example}
	\label{ex:cycle}
	Let 
	\begin{align*}
	A = 
	\begin{bmatrix}
	0 & I_{n-1} \\
	1 & 0
	\end{bmatrix} \in \mathbb{R}^{n \times n},\quad
	C = 
	\begin{bmatrix}
	1 & 0 & \cdots & 0
	\end{bmatrix} \in \mathbb{R}^{1 \times n}.
	\end{align*}
	The system is observable. However, if
	the DoS attacks are periodically launched at times
	$k=0,n,2n,\dots$, then we cannot construct state bounds, by
	using even an infinite number of measurements. In fact,
	$O\left( \{s_m\}_{m=0}^{\chi} \right)$ is not full column rank for 
	every set of time-steps $\{s_m\}_{m=0}^{\chi} $ satisfying
	$\{s_m\}_{m=0}^{\chi} \cap \{\ell n:\ell \in \mathbb{Z}_+\} = \emptyset$.
\end{example}

In
the following subsections, we see that  the condition (C2) is satisfied 
under certain assumptions on DoS.

\subsection{Sufficient condition for (C2) to hold}
\subsubsection{Approach to exploit periodic property of eigenvalues}
In 
Example \ref{ex:cycle}, $A$
is a circulant matrix, and
the eigenvalues of $A$ are given by 
$
e^{i 2\pi \frac{j}{n}}
$
($ j=0,\dots,n-1$),
all of which are on the unit circle.
In this subsection, exploiting this property of the eigenvalues,
we provide a sufficient condition for the generalized observability matrix $O\left( \{s_m\}_{m=0}^{\chi} \right)$
to be full column rank.
Assumptions on the plant are given as follows:
\begin{assumption}[Periodicity of eigenvalues]
	\label{assump:diagonal}
	For every $j=1,\dots,Q$,
	let $\zeta_j \in \mathbb{N}$ be one or a prime number, $\lambda_j \in \mathbb{C}$ be nonzero, and
	$a_{j,1},\dots,a_{j,n_j} \in \mathbb{Z}$ satisfy \mbox{$a_{j,\ell_1} \not\equiv a_{j,\ell_2}$ (mod $\zeta_j$)} for all
	$\ell_1,\ell_2 = 1,\dots,n_j$ with $\ell_1 \not= \ell_2$.
	The matrix $A$ is similar to a diagonal matrix $\Lambda := {\rm diag}(\Lambda_1,\dots,\Lambda_Q)$,
	where 
	\[
	\Lambda_j := {\rm diag} \left(\lambda_j e^{i 2\pi \frac{a_{j,1}}{\zeta_j}},\dots, \lambda_j e^{i 2\pi \frac{a_{j,n_j}}{\zeta_j}} \right)
	\quad \forall j \in 1,\dots,Q.
	\]
	Moreover, $(\lambda_{j_1} / \lambda_{j_2})^k \not =1$ for every $k \in \mathbb{N}$ and for every 
	$j_1,j_2 = 1,\dots,Q$ with $j_1 \not=j_2$.
\end{assumption}

\begin{assumption}[Single-output  system]
	\label{assump:single_output}
	The plant is a single-output system, that is,
	$C \in \mathbb{R}^{1\times n_x}$.
\end{assumption}

The following theorem shows that 
if the assumptions above are satisfied and if
the DoS duration bound $\nu_d$ is sufficiently small, then
the condition (C2) holds in finite time.
\begin{theorem}
	\label{thm:prime}
	Suppose that Assumptions \ref{assump:duration}, 
	\ref{assump:observability}, 
	\ref{assump:initial_time}, \ref{assump:diagonal}, and \ref{assump:single_output} hold.
	Let $\zeta \in \mathbb{N}$ be the least common multiple of $\zeta_1,\dots,\zeta_Q$.
	If the DoS duration bound $\nu_d$ satisfies
	\begin{equation}
	\label{eq:duration_prime}
	\nu_d  < \frac{1}{\zeta}
	\end{equation}
	then $O\left( \{s_m\}_{m=0}^{\chi} \right)$ defined in \eqref{eq:Obs_matrix_DoS} 
	is full column rank
	by time $k=(\ell_e + 1)\zeta$, where $\ell_e \in \mathbb{Z}_+$ is the maximum integer satisfying
	\begin{equation}
	\label{eq:ell_cond_prime}
	\ell_e \leq \frac{\Pi_d + Z(A,C)}{1-\zeta\nu_d}
	\end{equation}
	for some constant $Z(A,C)\in \mathbb{Z}_+$ that depends only on $A,C$.
\end{theorem}

We prove this result, using the techniques developed in
\cite{Park2011, Rohr2014}; see the appendix for details.

\begin{remark}
	Apply Theorem \ref{thm:prime} to the case where
	$A$ is diagonalizable and has only positive real eigenvalues with multiplicity 1.
	Since the least common multiple $\zeta$ is one,
	the matrix $O\left( \{s_m\}_{m=0}^{\chi} \right)$ is full column rank
	in finite time for every $\nu_d \in [0,1)$ under Assumptions \ref{assump:duration}, 
	\ref{assump:observability}, and \ref{assump:single_output}.
	The same result follows also from Theorem 1 in \cite{Wang2011}.
\end{remark}

\begin{remark}
	Lemma \ref{lem:one_cycle} in the appendix shows that 
	if 
	$\Lambda$ consists of only one block, i.e.,
	$\Lambda = \Lambda_1$, then the conditions \eqref{eq:duration_prime} and 
	\eqref{eq:ell_cond_prime}
	can be replaced by less conservative conditions:
	\[
	\nu_d < \frac{\zeta_1 - n_1+1}{\zeta_1},\quad
	\ell_e \leq \frac{\Pi_d}{\zeta_1 - n_1+1 - \zeta_1 \nu_d}.
	\]
	\vspace{0pt}
\end{remark}

\subsubsection{Coders using only consecutive data}
If the decoder receives $\eta$ consecutive data, where $\eta$ is the observability index of $(C,A)$,
then we can obtain a state bound as in the case without DoS.
Here we design coders that 
construct state bounds
from $\eta$ consecutive data.
In this case, the time-steps $s_0,\dots,s_{\chi}$ in Section \ref{sec:basic_zooming_out}
are given by $\chi = \eta-1$ and
$s_m = s_0 + m$ for all $m=1,\dots,\chi$. 
The advantage of this approach over Theorem~\ref{thm:prime} 
is that it is also applicable to multi-output systems.
\begin{theorem}
	\label{thm:DoS_frequency_duration}
	Let Assumptions \ref{assump:duration}, \ref{assump:freq}, \ref{assump:observability}, and \ref{assump:initial_time} hold.
	If the DoS duration and frequency bounds $\nu_d$ and $\nu_f$ satisfy
	\begin{equation}
	\label{eq:duration_frequency_cond}
	\nu_d  < 1 - (\eta-1) \nu_f,
	\end{equation}
	then the decoder receives $\eta$ consecutive data
	by time $k=k_e+1$, where $k_e \in \mathbb{Z}_+$ is
	the maximum integer satisfying
	\begin{equation}
	\label{eq:T_cond}
	k_e \leq \frac{\Pi_d + (\Pi_f+1)(\eta-1)}{1 - \nu_d + (\eta-1) \nu_f}.
	\end{equation}
	\vspace{0pt}
\end{theorem}
\begin{proof}
	Choose $k_e \in \mathbb{N}$ arbitrarily.
	Let DoS attacks occur at 
	\[
	k = k_m,\dots,k_m+\tau_m-1 \qquad \forall m=1,\dots,p
	\]
	on the interval $[0,k_e)$, where $k_m \in \mathbb{Z}_+$, $\tau_m \in \mathbb{N}$ for
	every $m=1,\dots,p$ and
	\begin{gather*}
	k_m+\tau_m < k_{m+1}\qquad \forall m=1,\dots,p-1.
	\end{gather*}
	In other words, $k_m$ and $\tau_m$ denote 
	the beginning time and the length of $m$th DoS interval.
	By definition, 
	$\sum_{m=1}^p \tau_m = \Phi_d(k_e)$ and $p = \Phi_f(k_e)$.
	We also define $k_0 := 0$, $\tau_0 := 0$, and $k_{p+1} := k_e$.

	Assume, to reach a contradiction, that 
	\begin{equation}
	\label{eq:contradiction_assump}
	k_{m+1} - k_m - \tau_m \leq \eta -1\qquad
	\forall m=0,\dots,p,
	\end{equation}
	which implies that 
	the decoder  receives at most 
	$\eta-1$ consecutive data on the interval $[0,k_e)$.
	Applying induction to \eqref{eq:contradiction_assump}, 
	we obtain
	\begin{align*}
	k_e = k_{p+1} 
	&\leq k_p + \tau_p + \eta - 1 \\
	&\leq \cdots 
	\leq \sum_{m=0}^p \tau_m + (p+1)(\eta-1).
	\end{align*}
	From Assumptions \ref{assump:duration} and \ref{assump:freq},
	it follows that
	\[
	k_e \leq (\Pi_d + \nu_d k_e) + (\Pi_f + \nu_f k_e + 1) (\eta-1),
	\]
	and hence
	\begin{equation}
	\label{eqT_cond}
	(1 - \nu_d - (\eta-1) \nu_f)k_e \leq \Pi_d + (\Pi_f +1)(\eta-1).
	\end{equation}
	Since $k_e \in \mathbb{N}$ was arbitrary, 
	the condition \eqref{eq:duration_frequency_cond}
	leads to a contradiction for a sufficiently large $k_e > 0$.
	Moreover, 
	if \eqref{eq:duration_frequency_cond} holds,
	then $k_e$ must satisfy the inequality \eqref{eq:T_cond}.
	In other words, if $k_e$ is the maximum integer satisfying the inequality \eqref{eq:T_cond},
	then $\eta$ consecutive data are transmitted successfully by time $k=k_e+1$.
	This completes the proof.
\end{proof}

In the case $\nu_d < \nu_f$, 
namely, when DoS attacks are frequently launched,
the following proposition is also useful.
\begin{proposition}
	\label{thm:DoS_duration}
	Suppose that Assumptions \ref{assump:duration}, \ref{assump:observability}, and \ref{assump:initial_time} hold.
	If the DoS duration bound $\nu_d$  satisfies
	\begin{equation}
	\label{eq:duration_cond}
	\nu_d < \frac{1}{\eta},
	\end{equation}
	then the decoder receives $\eta$ consecutive data by time $k=k_e+1$, where $k_e \in \mathbb{Z}_+$ is
	the maximum integer satisfying
	\begin{equation}
	\label{eq:T_cond_duration_only}
	k_e \leq \frac{(\Pi_d +1)\eta - 1}{1 - \eta \nu_d}.
	\end{equation}
	\vspace{0pt}
\end{proposition}
\begin{proof}
	Choose $k_e \in \mathbb{N}$ arbitrarily and 
	let DoS attacks occur at 
	\[
	k = t_m \in \mathbb{Z}_+\qquad \forall m=1,\dots,p.
	\]
	on the interval $[0,k_e)$, where
	$
	t_m < t_{m+1}
	$	
	for all $m=1,\dots,p$.
	Then $p = \Phi_d(k_e)$ by definition.
	
	Define $t_{p+1} := k_e$.
	Assume, to get a contradiction, that 
	\[
	t_1 \leq \eta - 1,\quad
	t_{m+1} - t_m \leq \eta \quad \forall m=1,\dots,p-1.
	\]
	Then we obtain
	\begin{align*}
	k_e = t_{p+1} \leq t_p + \eta \leq \cdots \leq (p+1)\eta - 1.
	\end{align*}
	By Assumptions \ref{assump:duration},
	$
	k_e \leq (\Pi_d+\nu_dk_e + 1) \eta - 1,
	$
	and hence
	\[
	(1-\eta\nu_d)k_e \leq (\Pi_d+1) \eta - 1.
	\]
	The rest of the proof
	follows the same
	lines as that of Theorem \ref{thm:DoS_frequency_duration}.
	Therefore, we shall omit it.
\end{proof}

\begin{remark}
	Suppose that
	the encoder at the plant side redundantly sends the set of output data, 
	$Q_{k-\eta+1}(y_{k-\eta+1}),\dots, Q_k(y_{k})$, at every time 
	$k \geq \eta-1$.
	The decoder can obtain an initial state bound  if a data set 
	whose values are all zero
	is successfully transmitted.
	Thus, the zooming-out procedure by this redundant scheme finishes in finite time
	for all $\nu_d \in [0,1)$.
\end{remark}

\subsubsection{With bounded lengths in DoS periods}
We here place the following assumption:
\begin{assumption}[Bounded length of DoS period]
	\label{assump:consecutiveDOS}
	For a given $\varpi \in \mathbb{N}$,
	at most $\varpi$
	consecutive DoS attacks occur.
\end{assumption}

If the DoS duration condition \eqref{eq:stDOSduration} holds,
then $\varpi \in \mathbb{N}$ has to satisfy
$
\varpi \leq \Pi_d + \nu_d \varpi.
$
Therefore, 
$
\varpi \leq \Pi_d/(1-\nu_d).
$

For packet losses including DoS attacks in Assumption \ref{assump:consecutiveDOS},
the earlier study \cite{Jungers2018} shows that 
we can check in finite time whether or not there exists $k \in \mathbb{Z}_+$ such that
the matrix $O_{\sigma}(k) $ is full column rank, where
the binary function $\sigma:\mathbb{Z_+} \to \{0,1\}$ is defined by
\[
\sigma(k) :=
\begin{cases}
0 & \text{if packet loss occurs at $k$} \\
1 & \text{if packet loss does not occur at $k$},
\end{cases}
\]
and
\[
O_{\sigma}(k) :=
\begin{bmatrix}
\sigma(0) C \\
\sigma(1) CA \\
\vdots \\
\sigma(k-1)CA^{k-1}
\end{bmatrix}.
\]
Moreover, if such $k$ (the time when $O_{\sigma}(k) $ is full column rank) exists, then 
$k$ is upper-bounded by
a certain value $k_e \in \mathbb{Z}_+$ 
that depends only on $\varpi$ and $(C,A)$; see Proposition 1, Theorem 2, and Remark 2 in \cite{Jungers2018}.
When we 
regard $s_0 \in \mathbb{Z}_+$ as $s_0 = 0$ in \eqref{eq:Obs_matrix_DoS},
$O_{\sigma}(k)$ is full column rank if and only if $O\left( \{s_m\}_{m=0}^{\chi} \right)$
is full column rank.
Therefore,
we can immediately apply the result in \cite{Jungers2018}. 
Thus, we can find in finite time whether or not
a full column rank $O\left( \{s_m\}_{m=0}^{\chi} \right)$ exists. In addition, if it exists, then
the decoder can construct a state bound by a certain time that depends on $\varpi$ and $(C,A)$.

\subsection{Lyapunov stability}
\label{sec:Lyap_stability}
Combining the encoding schemes in Section~\ref{sec:exp_conv} and this section,
we achieve Lyapunov stability.
\begin{definition}[Lyapunov stability]
	The feedback system in Section \ref{sec:closed_loop}
	achieves Lyapunov stability if
	for every $\varepsilon > 0$, there exists $\delta >0$ such that 
	\begin{equation}
	\label{eq:Lyap_cond}
	|x_0|_{\infty} < \delta \quad \Rightarrow \quad 
	|x_k|_{\infty},~|\hat x_k|_{\infty} < \varepsilon \quad \forall k \in \mathbb{Z}_+.
	\end{equation}
	\vspace{-8pt}
\end{definition}

In the following theorem, we use
Theorems \ref{thm:state_conv} and \ref{thm:DoS_frequency_duration}, but
similar results can be obtained from other combinations such as
Corollary \ref{coro:simle_case} and Theorem \ref{thm:prime}.
\begin{theorem}
	Suppose that Assumptions \ref{assump:duration}, \ref{assump:freq}, \ref{assump:stabilizability_detectability},
	\ref{assump:observability}, and \ref{assump:odd} hold.
	If the quantization level $N$ and the DoS duration and frequency bounds $\nu_d$, $\nu_f$ satisfy
	\eqref{eq:with_frequency_Nnu} and 
	\eqref{eq:duration_frequency_cond}, 
	then the closed-loop system achieves Lyapunov stability under
	the encoding scheme in Sections~\ref{sec:basic_zooming_in}, 
	\ref{sec:Duration_Frequency}, and \ref{sec:basic_zooming_out}. 
\end{theorem}
\begin{proof}
	Suppose that $\delta>0$ satisfies 
	\begin{equation}
	\label{eq:delta_cond1}
	\delta < E_0^x.
	\end{equation} 
	Then $Q_k(y_k) = 0$ for every time $k$
	during the zooming-out procedure.
	Let a state bound be obtained at time $k=T_1$,
	namely, let $s_\chi+1$ in Section \ref{sec:basic_zooming_out} be equal to $T_1$.
	By Theorem \ref{thm:DoS_frequency_duration},
	$T_1$ has a certain upper bound $\overline T_1 \in \mathbb{Z}_+$. Hence,
	for the error bound $E_{T_1}$ defined as in Section \ref{sec:basic_zooming_out},
	there exists $\overline E >0$ such that $E_{T_1} \leq \overline E$.
	
	Theorem \ref{thm:state_conv} shows that,
	in the zooming-in stage, there exist $\Omega \geq 1$ and $\gamma \in (0,1)$ such that 
	$x_k$ and $\hat x_k$ satisfy
	\[
	|x_k|_{\infty},~|\hat x_k|_{\infty} \leq \Omega \gamma^{k-T_1} E_{T_1}\qquad \forall k \geq T_1.
	\]
	We set an integer $T_2 \geq \overline T_1$ so that
	\begin{equation}
	\label{eq:T2_larger}
	\Omega \gamma^{T_2 - \overline T_1} \overline E< \varepsilon.
	\end{equation}
	Then
	$
	|x_k|_{\infty},~|\hat x_k|_{\infty}  < \varepsilon
	$
	for every $k \geq T_2$.
	
	Let us next show that 
	\begin{equation}
	\label{eq:T2_smaller}
	|x_k|_{\infty},~|\hat x_k|_{\infty}   < \varepsilon \qquad \forall k \leq T_2.
	\end{equation}
	Define 
	\begin{align*}
	\Upsilon :=\min &\big\{
	\big\|A^{s_{\chi} - s_0 + 1}  O\left( \{s_m\}_{m=0}^{p} \right)^{\dagger} \big\|_{\infty}:
	O\left( \{s_m\}_{m=0}^{\chi}  \right)\text{~is } \\
	&\hspace{-5pt}\quad \text{full column rank and~} 
	\{s_m\}_{m=0}^{\chi} \subset \{0,\dots,\overline{T}_1\} \big\}.
	\end{align*}
	By the definition of $E_{T_1}$ in the zooming-out procedure and 
	the update rule of $E_{R,k}$ in the zooming-in procedure, 
	we obtain 
	\[ E_{R,k} \geq \Upsilon \|C\|_{\infty}  \cdot \|R\|_{\infty} \theta^{T_2} E_0^x\qquad \forall k \in [T_1, T_2].
	\]
	For each $k \in (T_1, T_2]$, 
	if $q_{\ell} = 0$ for every $\ell \in [T_1, k)$, then $u_k = 0$, and hence
	$
	|y_{k}|_{\infty} \leq \|C\|_{\infty} \cdot \|A\|_{\infty}^{T_2} \delta.
	$
	Moreover, in such a case, if 
	\[
	|y_k - \hat y_k|_{\infty} = 
	|y_k|_{\infty} \leq \frac{\|CR^{-1}\|_{\infty}E_{R,k}}{N},
	\]
	then $q_k = 0$,
	because the number of quantization levels $N$ is odd.
	Therefore, if $\delta>0$ satisfies
	\begin{equation}
	\label{eq:delta_cond2}
	\|A\|^{T_2}\delta < \frac{\Upsilon \|CR^{-1}\|_{\infty}\cdot  \|R\|_{\infty}\theta^{T_2}E_0^x}{N},
	\end{equation}
	then $q_k = 0$ for every $k \in [T_1, T_2]$.
	Hence $u_k = 0$ for every $k \leq T_2$.
	Thus, if $\delta>0$ additionally satisfies
	\begin{equation}
	\label{eq:delta_cond3}
	\|A\|^{T_2}\delta < \varepsilon,
	\end{equation}
	then \eqref{eq:T2_smaller} holds.
	
	In summary, 
	if $\delta > 0$ satisfies \eqref{eq:delta_cond1}, \eqref{eq:delta_cond2}, 
	and \eqref{eq:delta_cond3}, then 
	\eqref{eq:Lyap_cond} holds. Thus,
	Lyapunov stability is achieved.
\end{proof}

\section{Numerical Examples}
\subsection{Plant and controller}
A linearized model of the unstable batch reactor 
studied in \cite{Rosenbrock1972} is given by
$\dot x(t) = A_cx(t) + B_cu(t)$ and $
y(t) = C_cx(t),
$
where
\begin{align*}
A_c &:= \begin{bmatrix}
1.38 & -0.2077 & 6.715 & -5.676 \\
-0.5814 & -4.29 & 0 & 0.675 \\
1.067 & 4.273 & -6.654 & 5.893 \\
0.048 & 4.273 & -1.343 & -2.104
\end{bmatrix} \\
B_c &:= 
\begin{bmatrix}
0 & 0 \\
5.679 & 0 \\
1.136 & -3.146 \\
1.136 & 0
\end{bmatrix},\quad 
C_c := 
\begin{bmatrix}
1 & 0 & 1 & -1 \\
0 & 1 & 0 & 0
\end{bmatrix}.
\end{align*}
Here we discretize this plant with the sampling period $h = 0.2$.
We use the feedback gain $K$ that is the
linear quadratic regulator whose state weighting matrix and
input weighting matrix are the identity matrices $I_4$ and $I_2$, respectively.
The observer gain $L$ is given by 
the gain of the steady-state Kalman filter whose
covariances of the process noise and measurement noise are $I_4$ and $0.1\times I_2$, respectively.

\subsection{Relationship between quantization level and DoS duration and frequency}
By Corollary \ref{coro:simle_case}, we obtain a relationship between the quantization level $N$
and the DoS duration bound $\nu_d$ for the state convergence.
Here we choose the matrix $R$  so that 
\begin{equation}
\label{eq:R_choice_ex1}
\|R(A-LC)R^{-1}\|_{\infty} = \varrho(A-LC).
\end{equation}
Each circle in Fig.~\ref{fig:Nnud} illustrates the minimum
integer $N$ satisfying \eqref{eq:DoS_cond_coro} in 
Corollary \ref{coro:simle_case}.
This corollary shows that as
the quantization level $N$ increases to infinity, 
the DoS duration bound $\nu_d$ goes to 
\[
\frac{-  \log \|R(A-LC)R^{-1}\|_{\infty}}{\log \|RAR^{-1}\|_{\infty} -  \log \|R(A-LC)R^{-1}\|_{\infty}}
\approx 0.1405.
\]

\begin{figure}[tb]
	\centering
	\includegraphics[width = 7cm]{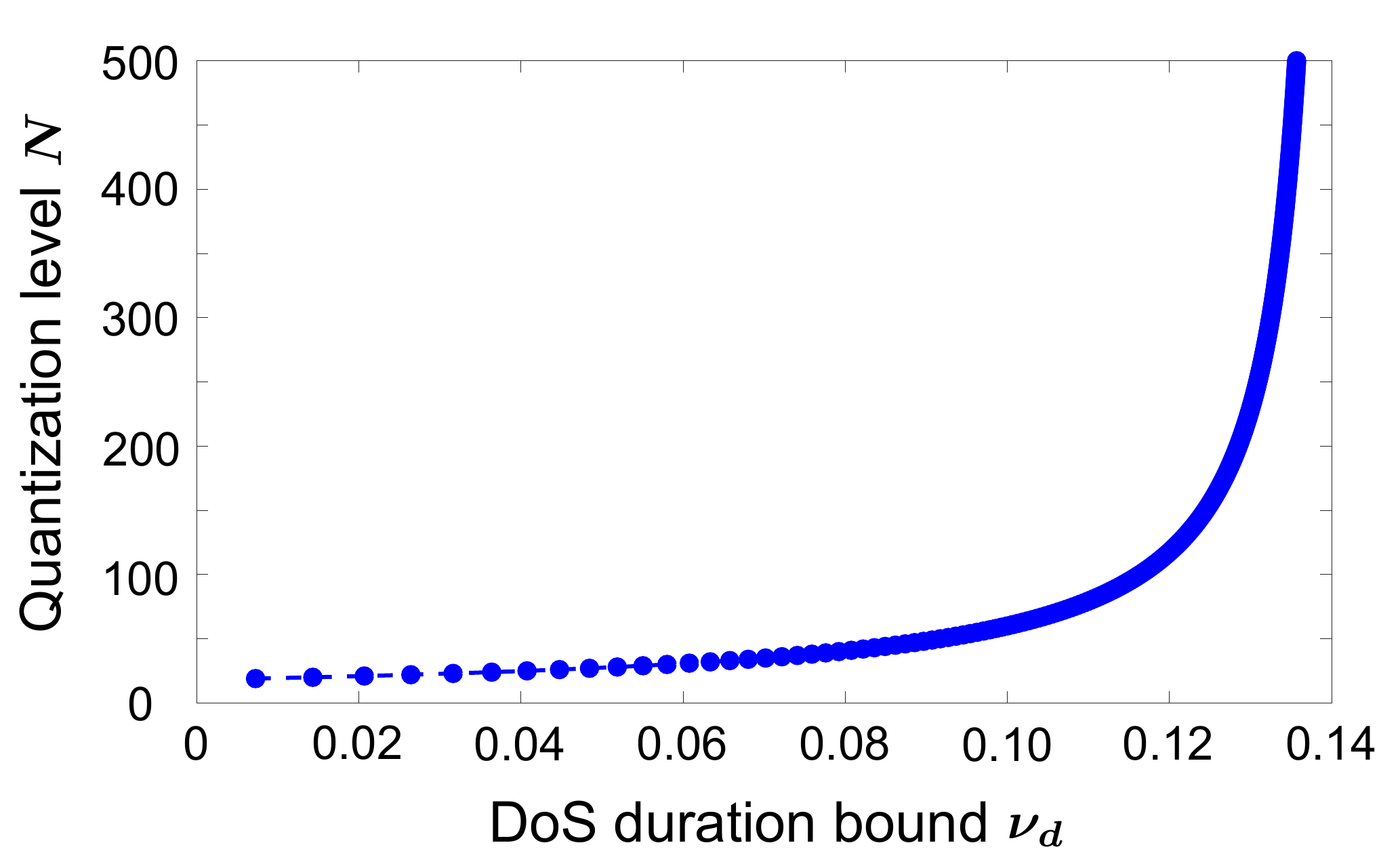}
	\caption{Relationship between quantization level $N$ and DoS duration bound $\nu_d$ under
		encoding scheme with center at output estimate.}
	\label{fig:Nnud}
	\vspace{-8pt}
\end{figure}

Let us next see a relationship between
the quantization level $N$ and the DoS duration and frequency bounds $\nu_d$, $\nu_f$.
In Theorem \ref{thm:state_conv},
we choose the matrix $R$  so that 
\begin{equation}
\label{eq:R_choice_ex2}
\|RAR^{-1}\|_{\infty} = \varrho(A).
\end{equation}
The surface in Fig.~\ref{fig:3d_plot} depicts 
the minimum integer satisfying \eqref{eq:DoS_cond} in 
Theorem~\ref{thm:state_conv}
for given DoS duration and frequency bounds.
Using Theorem~\ref{thm:state_conv}, we find that  as 
the quantization level $N$ goes to infinity,
the DoS duration and frequency bounds $\nu_d, \nu_f$ get close to the line
\[
\nu_d = 
\frac{-\nu_f  \log M_0 - \log \rho}{\log \|RAR^{-1}\|_{\infty} - \log \rho}
\approx -1.9080\nu_f + 0.3042.
\]
This can be also observed in Fig.~\ref{fig:cont_plot}.

\begin{figure}
	\centering
	\subcaptionbox{ 
		3D plot.
		\label{fig:3d_plot}}
	{\includegraphics[width = 6.7cm]{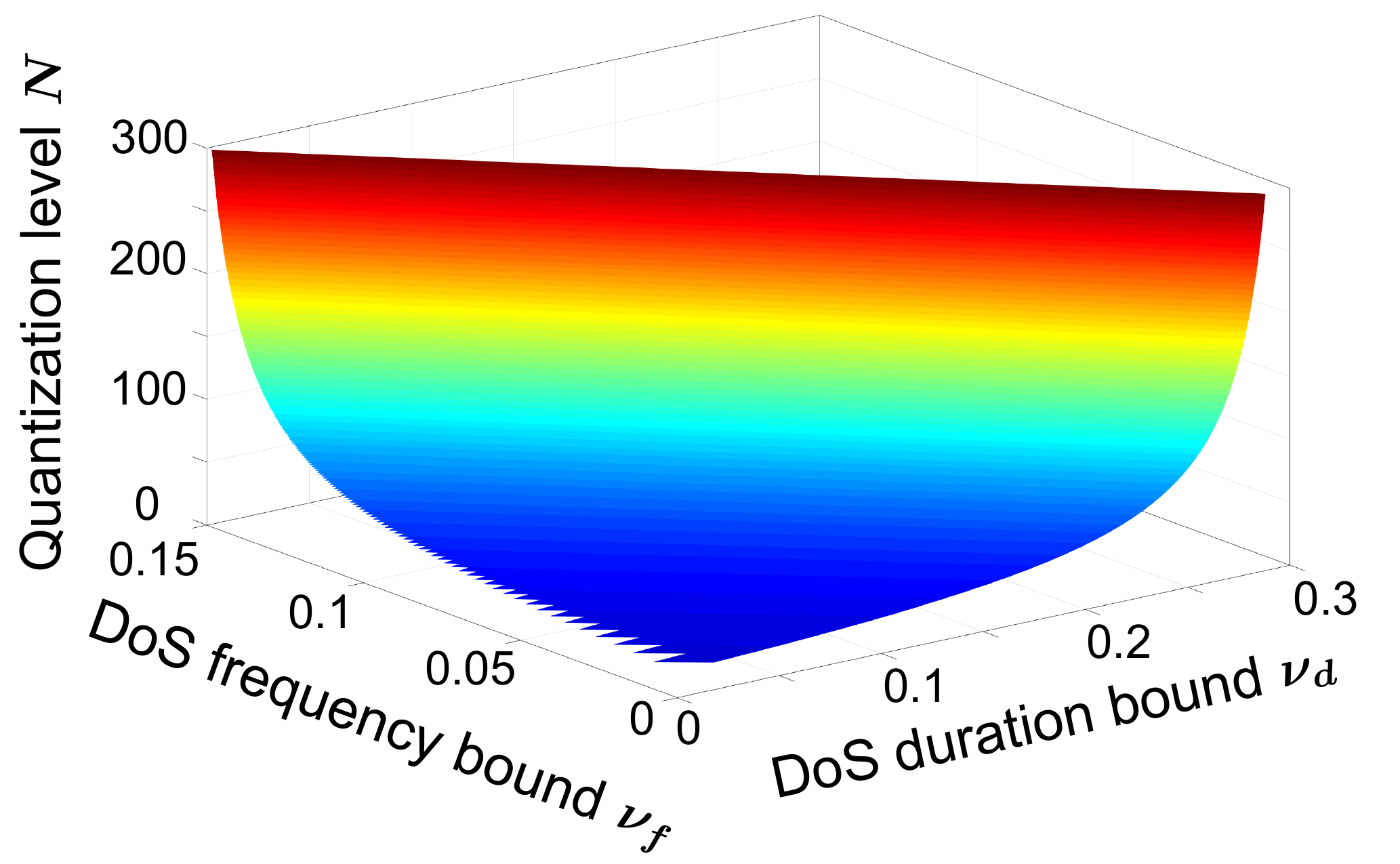}} \vspace{10pt}\\
	\subcaptionbox{Contour plot.
		\label{fig:cont_plot}}
	{\includegraphics[width = 6.7cm,clip]{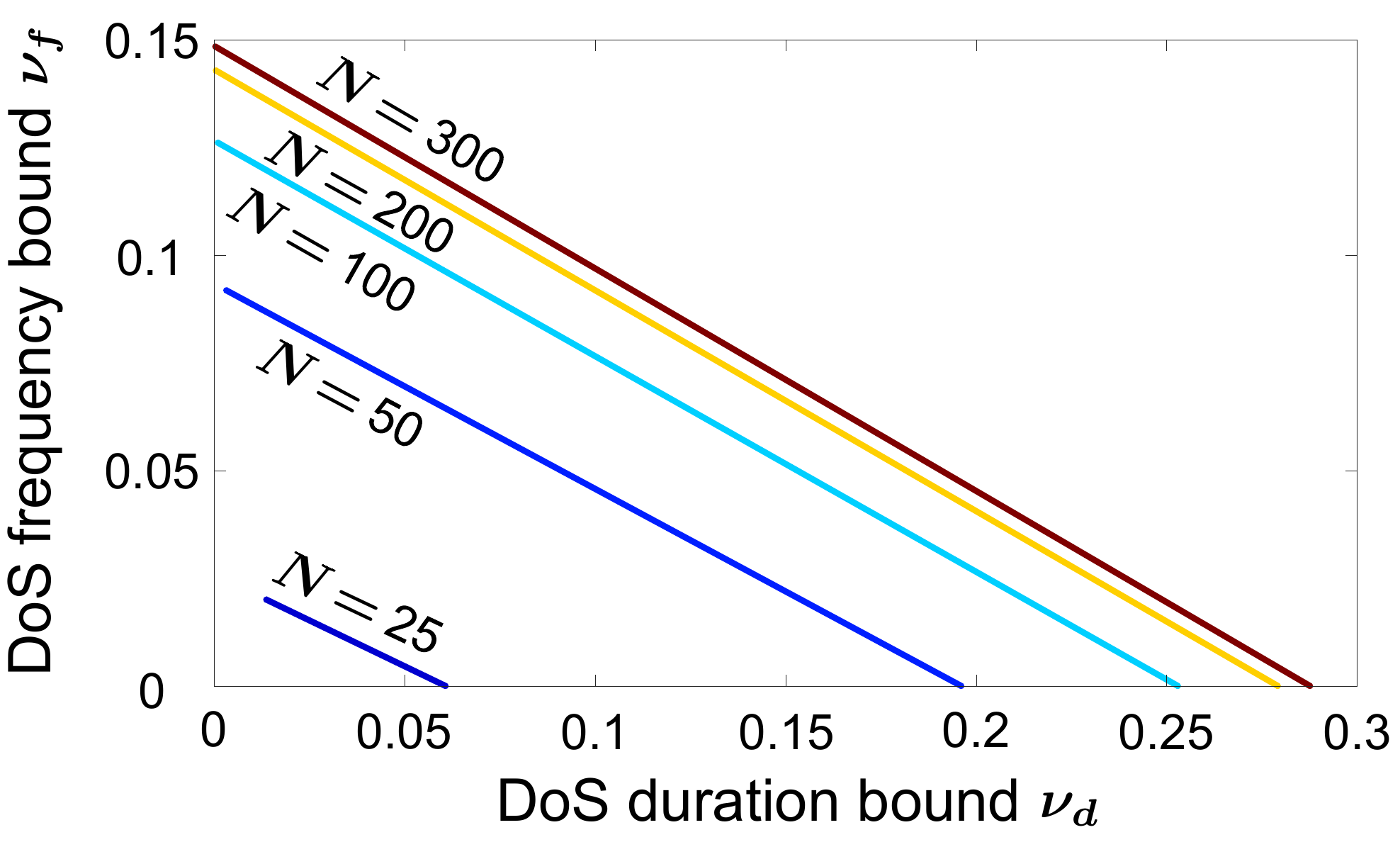}}
	\caption{Relationship between quantization level $N$ and DoS duration and frequency bounds $\nu_d,\nu_f$ under
		encoding scheme with center at output estimate.\label{fig:Nnudf}}
	\vspace{-8pt}
\end{figure}


Finally, we see a relationship between
the quantization level $N$ and the DoS duration bound $\nu_d$
for the encoding scheme whose quantization center is the origin.
Fig.~\ref{fig:Nnud_origin} illustrates the minimum integer $N$ satisfying
\eqref{eq:DoS_cond_coro_origin} in Corollary~\ref{coro:simle_case_origin},
where $R_{\text{cl}}$ is chosen so that
\[
\|R_{\text{cl}}A_{\text{cl}}R_{\text{cl}}^{-1}\|_{\infty} = \varrho(A_{\text{cl}}).
\]
By Corollary \ref{coro:simle_case_origin}, 
the DoS duration bound $\nu_d$ converges to
\[
\frac{-  \log \|R_{\text{cl}}A_{\text{cl}}R_{\text{cl}}^{-1}\|_{\infty}}{\log \|R_{\text{cl}}A_{\text{op}}R_{\text{cl}}^{-1}\|_{\infty} -  \log \|R_{\text{cl}}A_{\text{cl}}R_{\text{cl}}^{-1}\|_{\infty}}
\approx 0.0736.
\]
as the quantization level $N$ goes to infinity.
We can observe that
the encoder with center at the origin
needs more data rates in exchange for the reduction of computational resources of the coders.

\begin{figure}[tb]
	\centering
	\includegraphics[width = 7cm]{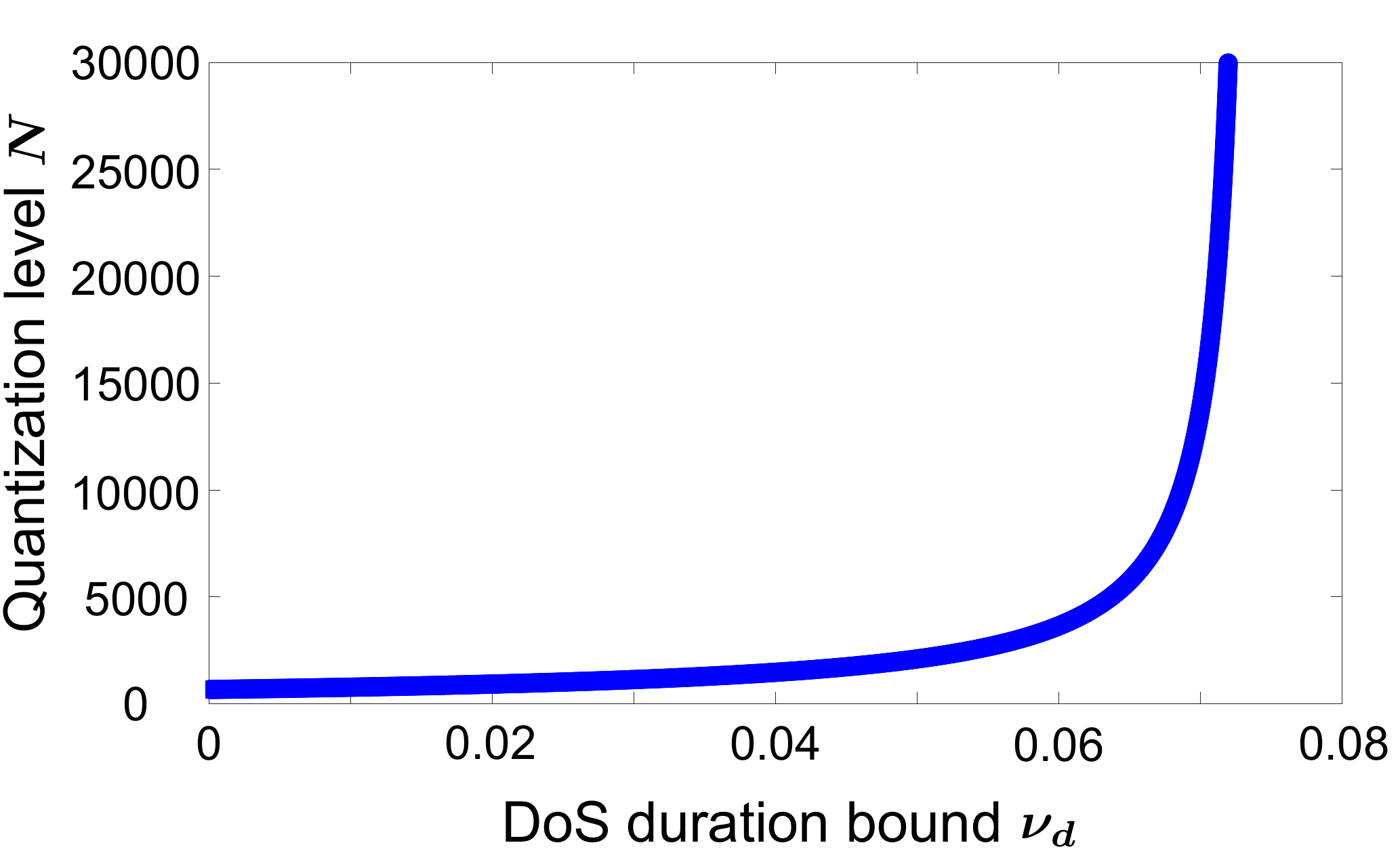}
	\caption{Relationship between quantization level $N$ and DoS duration bound $\nu_d$ under
		encoding scheme with center at origin.}
	\label{fig:Nnud_origin}
	\vspace{-8pt}
\end{figure}

\subsection{Time responses}
We present time responses under the encoding schemes with center at the output estimate.
Through simulation results, we see how conservative the obtained sufficient conditions are.
\subsubsection{DoS attacks}
We set the quantization level $N$  to be $N = 71$. 
The closed-loop system with the encoding scheme of Corollary \ref{coro:simle_case}
achieves exponential convergence under the DoS duration 
\begin{equation}
\label{eq:nu_d_cond_example}
\nu_d <  0.106,
\end{equation}
where
we construct the transformation matrix $R$ so that \eqref{eq:R_choice_ex1} holds.
Moreover,
Proposition \ref{thm:DoS_duration} shows that 
an initial state bound is obtained in finite time if 
$
\nu_d <  0.5.
$

On the other hand, 
the encoding scheme of Theorem \ref{thm:state_conv}
achieves exponential convergence if the DoS duration and frequency bounds $\nu_d$
and $\nu_f$ satisfy
\begin{equation}
\label{eq:nu_df_cond_example}
\nu_d < 0.230 - 2.041 \nu_f,
\end{equation}
where the matrix $R$ is chosen so that \eqref{eq:R_choice_ex2} is satisfied.
By  Theorem \ref{thm:DoS_frequency_duration},
the coders can construct an initial state bound in finite time if
$
\nu_d < 1 - \nu_f.
$

If the frequency of DoS attacks is sufficiently small,
then the encoding scheme of Theorem \ref{thm:state_conv} allows 
longer duration of DoS attacks without compromising the closed-loop stability in this example.
However,
the encoding scheme of Corollary \ref{coro:simle_case}
can tolerate DoS attacks with large frequency.
For instance, if $\nu_d = 0.1$ in \eqref{eq:nu_df_cond_example},
then $\nu_f < 0.063$. 
The encoding scheme of Corollary \ref{coro:simle_case}
allows the DoS attacks launched at times $k=10,20,30,\dots$, but
that of Theorem \ref{thm:state_conv} does not.

In the simulation below, 
we assume that the attacker knows all information on the closed-loop system,
which leads to effective DoS.
After a state bound is obtained, i.e., at the zooming-in stage,
DoS attacks occur if the following two conditions are both satisfied 
in addition to the above constraints on the duration and frequency:
\begin{equation}
\label{eq:DoS_rule}
|Ae_k|_{\infty} 
> \alpha_1 |e_k|_{\infty},\quad
|y_k - q_k|_{\infty} > \alpha _2 \frac{\|CR^{-1}\|_{\infty}}{N} E_{R,k},
\end{equation}
where $\alpha_1 = 1$ and $\alpha_2 = 1/2$.
Recall that
the maximum quantization error is given by
\[\frac{\|CR^{-1}\|_{\infty}}{N} E_{R,k}
\]
after a state bound is derived,
as shown in \eqref{eq:qe_y}.
As the constant $\alpha_1 \geq 0$ increases,
the estimation error $e_k$ becomes larger due to DoS attacks.
As the constant $\alpha_2 \in [0,1)$ becomes close to one, 
the second condition
leads to a larger quantization error $y_k - q_k$.
If $\alpha_1$ is too large or if $\alpha_2$ is too close to one,
then DoS attacks rarely occur.

\begin{remark}
	A more sophisticated design of DoS attacks was discussed in  Example 2.8 of \cite{Ahmet2018arXiv}, where
	the attacker decides whether to block data transmissions or not, by solving
	an optimization problem over a short horizon at each time
	like model predictive control (MPC).
	Compared with the rule based on \eqref{eq:DoS_rule},
	this MPC-like strategy requires computational resources because
	the attacker has to solve a 0-1 integer programming problem.
	However, DoS attacks can be effectively launched without
	tuning parameters.
\end{remark}

\subsubsection{Simulation results}
Let us denote the state $x_k$ and its estimate $\hat x_k$ by
$x = 
\begin{bmatrix}
x^1&x^2 &x^3&x^4
\end{bmatrix}^{\top}$ and 
$\hat x = 
\begin{bmatrix}
\hat x^1 & \hat x^2 & \hat x^3 & \hat x^4
\end{bmatrix}^{\top}$, respectively.
For the computation of time responses,
we set the initial state $x_0$  to be
$x_0 = 
\begin{bmatrix}
0 & 0.5 & 0.5 & 1
\end{bmatrix}^{\top}$.
The parameters $E_0$ and $\kappa$ for 
the encoding scheme to derive an initial state bound
are given by
$E_0 = 0.01$ and  $\kappa = 0.01$.

Figs.~\ref{fig:stable} and \ref{fig:unstable} show
time responses under the encoding scheme of Corollary \ref{coro:simle_case}
without assuming any frequency conditions of DoS attacks.
Fig.~\ref{fig:stable}
depicts the stable case where
DoS attacks satisfy the duration condition \eqref{eq:DOSduration} with 
$(\Pi_d, \nu_d) = (2,0.10)$.
DoS attacks occur on the intervals that are colored in gray.
Since 
the DoS duration $\nu_d = 0.10$ satisfies \eqref{eq:nu_d_cond_example}, 
the error bound $E_{R}$ exponentially decreases, which
leads to exponential convergence.

Fig.~\ref{fig:unstable} illustrates the unstable case of
the encoding scheme of Corollary \ref{coro:simle_case}
under DoS attacks with $(\Pi_d, \nu_d) = (2,0.11)$.
Since the DoS duration bound $\nu_d = 0.11$ does not satisfy \eqref{eq:nu_d_cond_example},  
DoS attacks make the error bound $E_{R}$
diverge in Fig.~\ref{fig:unstable_ER}. 
As the error bound $E_{R}$ increases, the worst-case  quantization error
becomes larger, which leads to the instability of the closed-loop system as shown in
Fig.~\ref{fig:unstable_x1}, although the difference between the 
threshold $0.106$ in \eqref{eq:nu_d_cond_example} and 
$\nu_d = 0.11$ used in Fig.~\ref{fig:unstable} is small. We see 
that the sufficient condition \eqref{eq:DoS_cond}
is fairly tight in this example.

Next
we compute time responses in stable and unstable cases under the encoding scheme of Theorem \ref{thm:state_conv},
assuming that the DoS duration and frequency are both averagely bounded.
By \eqref{eq:nu_df_cond_example}, if the DoS duration bound $\nu_d$ is given by $\nu_d= 0.15$,
then the frequency bound $\nu_f$ should satisfy $\nu_f < 0.0392$.
Fig.~\ref{fig:stable_freq} illustrates the stable case, where DoS attacks satisfy
the duration condition \eqref{eq:DOSduration} with  $(\Pi_d, \nu_d) = (2,0.15)$ and 
the frequency condition \eqref{eq:DOSfrequency} with $(\Pi_f,\nu_f) = (1,0.035)$. 
Fig.~\ref{fig:stable_x1_freq} shows that 
the closed-loop system achieves exponential convergence
despite longer DoS duration than in the case of Fig.~\ref{fig:stable}. 
This is because the encoding scheme in Theorem \ref{thm:state_conv}
has a small growth rate $\theta_a = 1.489$ in the presence of DoS attacks, compared with 
the growth rate $\vartheta_a = 2.901$ of the encoding scheme in Corollary \ref{coro:simle_case}.

Fig.~\ref{fig:unstable_freq} shows the time response in the unstable case, where
$(\Pi_d, \nu_d) = (2,0.15)$ and $(\Pi_f, \nu_f) = (1,0.045)$.
The DoS frequency bound $\nu_f$ is just slightly larger than the threshold $0.0392$, but
the error bound $E_R$ diverges. Consequently, the closed-loop system
is unstable.

In Figs.~\ref{fig:stable_x1_freq} and \ref{fig:unstable_x1_freq}, 
the trajectories of $x^1$ and $\hat x^1$
oscillate after DoS attacks, which is unique to the case with quantization.
These oscillations are caused by large
error bounds $E_{R}$ due to DoS attacks, as shown in 
Figs.~\ref{fig:stable_ER_freq} and \ref{fig:unstable_ER_freq}.
Even after DoS attacks, the quantized output $q$ is zero under coarse quantization until the error bound
becomes small. Hence
the observer does not estimate the plant state correctly.

\begin{figure}[tb]
	\centering
	\subcaptionbox{State $x^1$ and its estimate $\hat x^1$.
		\label{fig:stable_x1}}
	{\includegraphics[width = 7cm,clip]{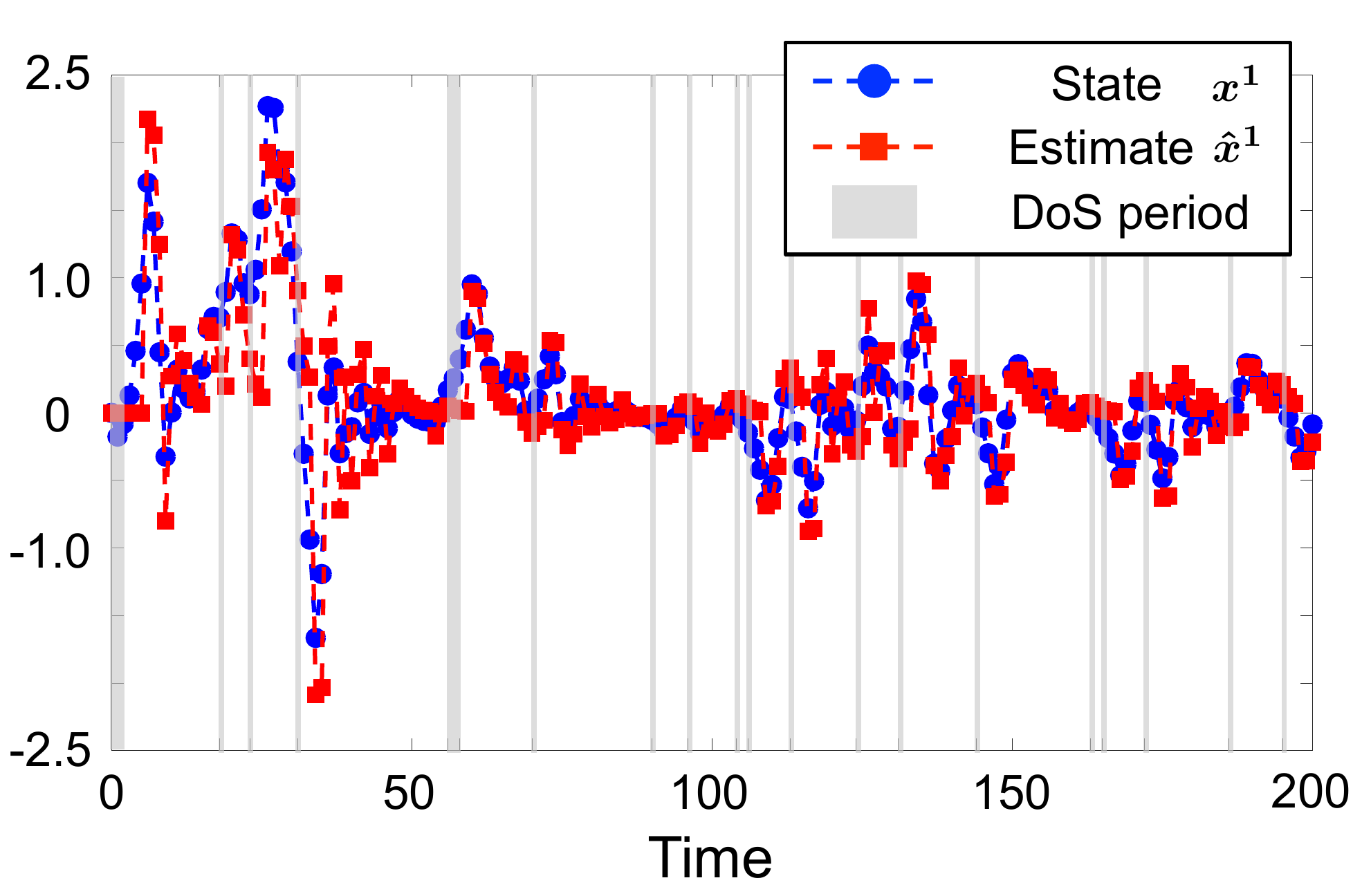}} \vspace{12pt}\\
	\subcaptionbox{Error bound $E_R$.
		\label{fig:stable_ER}}
	{\includegraphics[width = 7cm,clip]{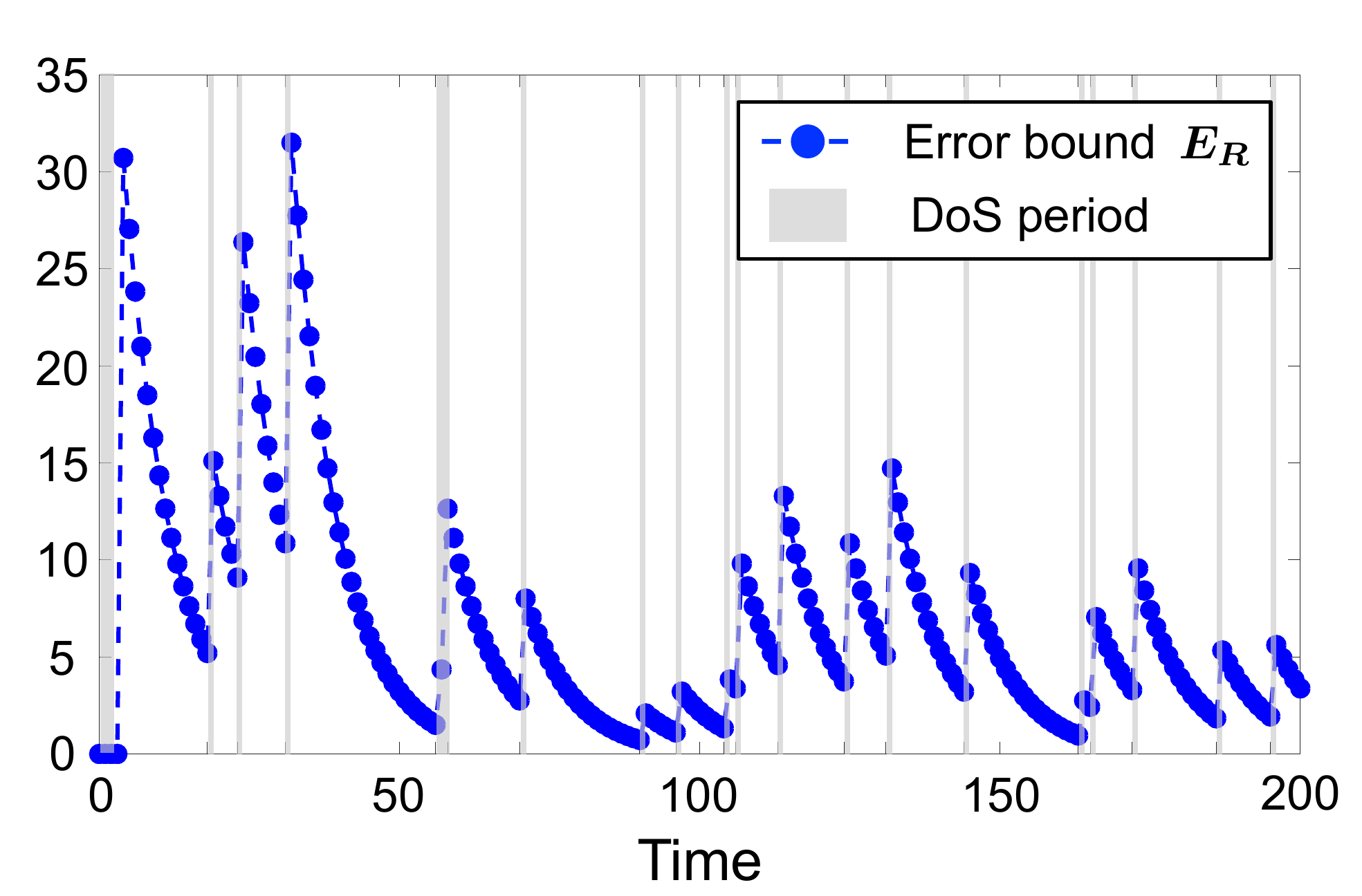}}
	\caption{Stable case without frequency condition ($\nu_d = 0.1$). \label{fig:stable}}
	\vspace{-8pt}
\end{figure}

\begin{figure}[tb]
	\centering
	\subcaptionbox{State $x^1$ and its estimate $\hat x^1$.
		\label{fig:unstable_x1}}
	{\includegraphics[width = 7cm,clip]{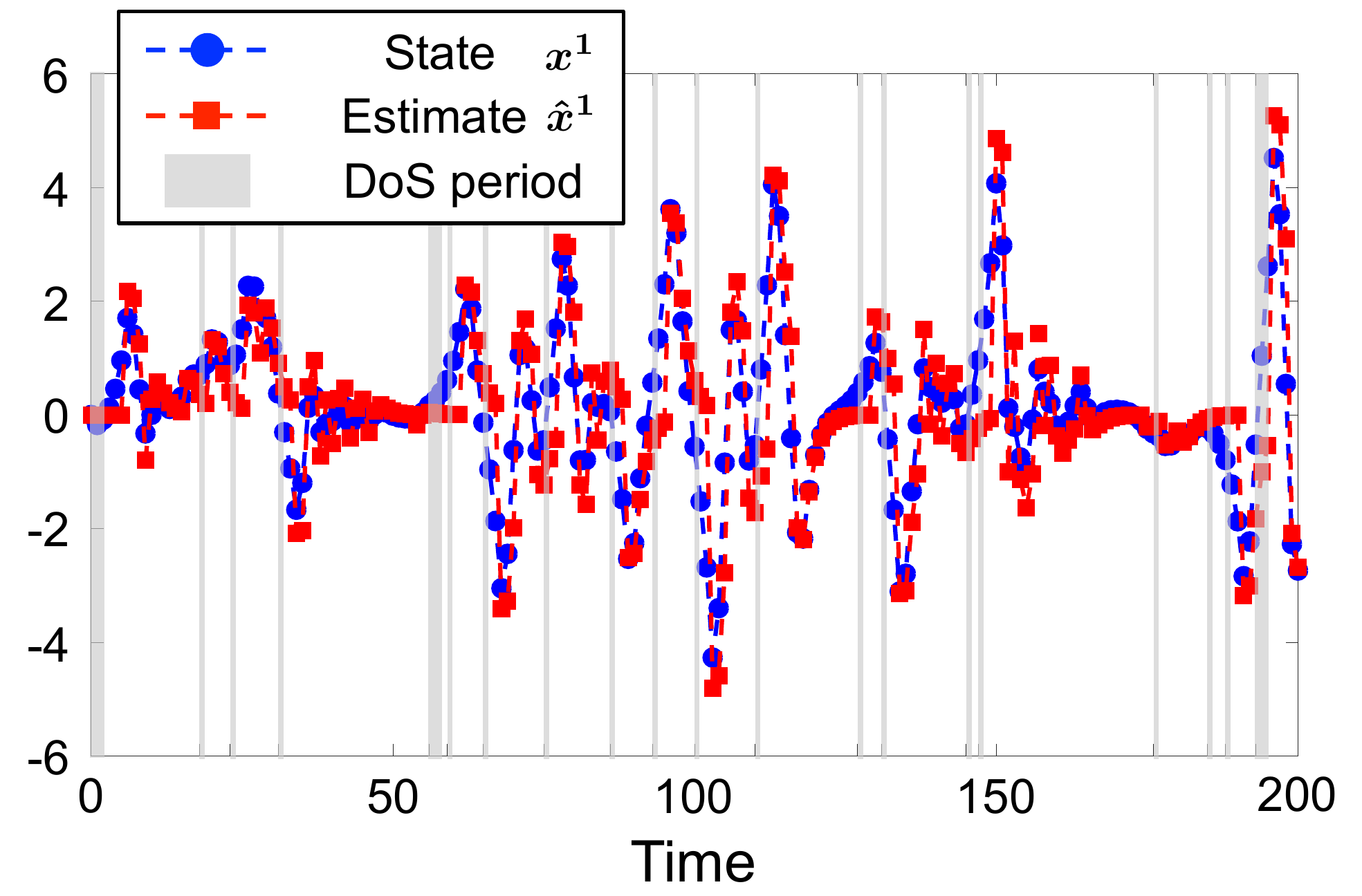}} \vspace{12pt}\\
	\subcaptionbox{Error bound $E_R$.
		\label{fig:unstable_ER}}
	{\includegraphics[width = 7cm,clip]{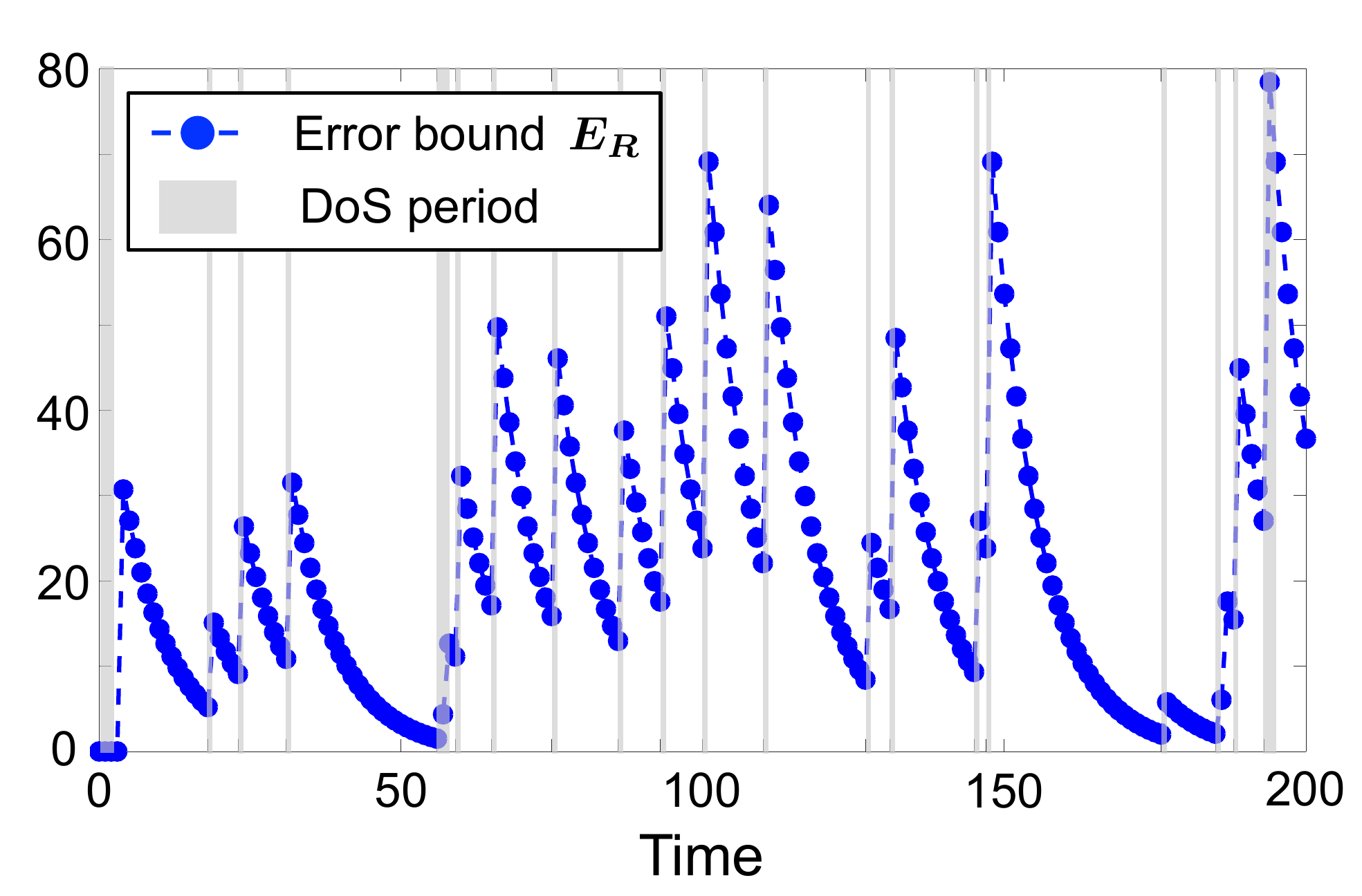}}
	\caption{Unstable case without frequency condition ($\nu_d = 0.11$). \label{fig:unstable}}
	\vspace{-8pt}
\end{figure}

\begin{figure}[tb]
	\centering
	\subcaptionbox{State $x^1$ and its estimate $\hat x^1$.
		\label{fig:stable_x1_freq}}
	{\includegraphics[width = 7cm,clip]{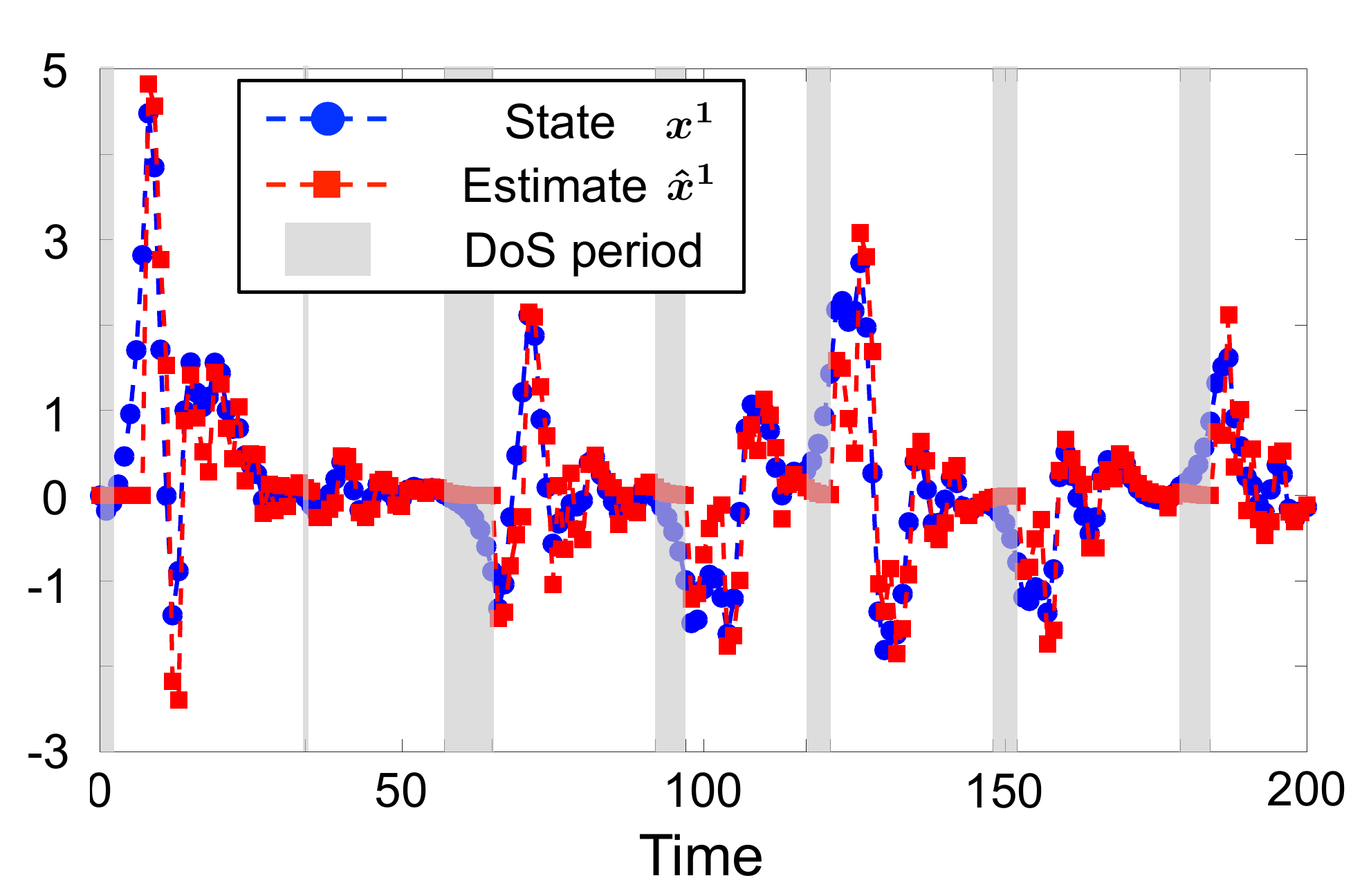}} \vspace{12pt}\\
	\subcaptionbox{Error bound $E_R$.
		\label{fig:stable_ER_freq}}
	{\includegraphics[width = 7cm,clip]{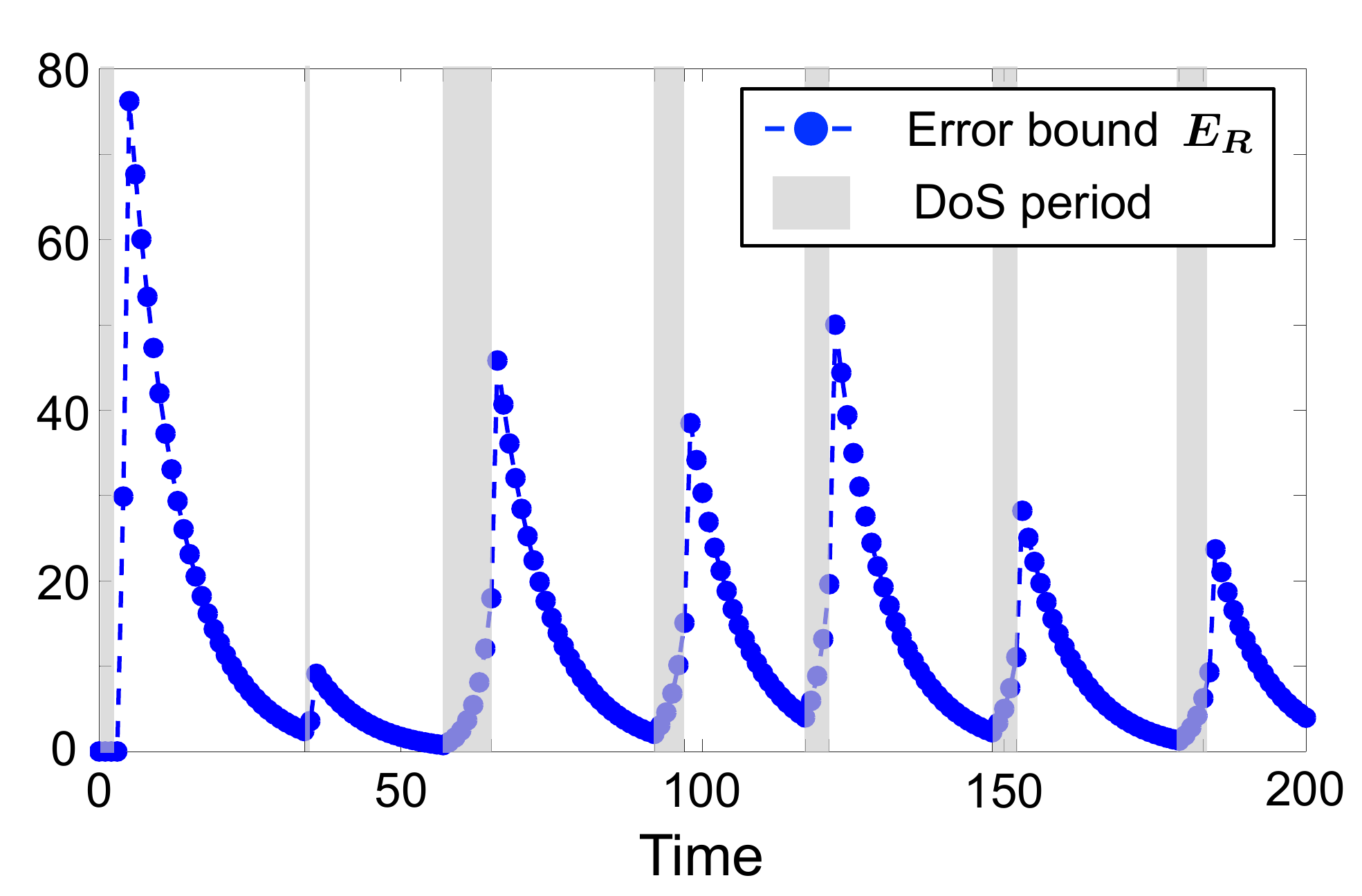}}
	\caption{Stable case with frequency condition ($\nu_d = 0.15,~\nu_f = 0.035$). \label{fig:stable_freq}}
	\vspace{-8pt}
\end{figure}

\begin{figure}[tb]
	\centering
	\subcaptionbox{State $x^1$ and its estimate $\hat x^1$.
		\label{fig:unstable_x1_freq}}
	{\includegraphics[width = 7cm,clip]{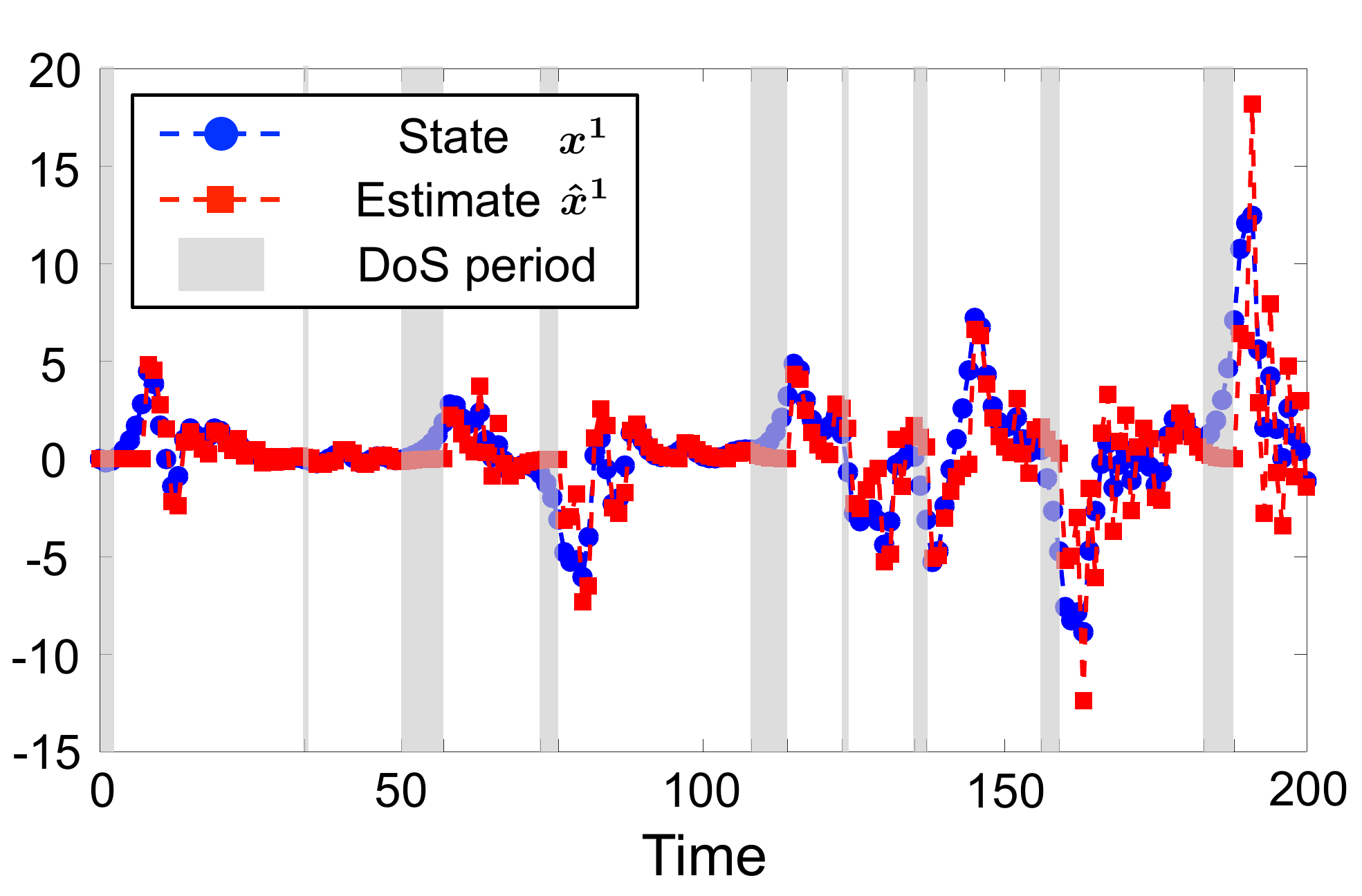}} \vspace{12pt}\\
	\subcaptionbox{Error bound $E_R$.
		\label{fig:unstable_ER_freq}}
	{\includegraphics[width = 7cm,clip]{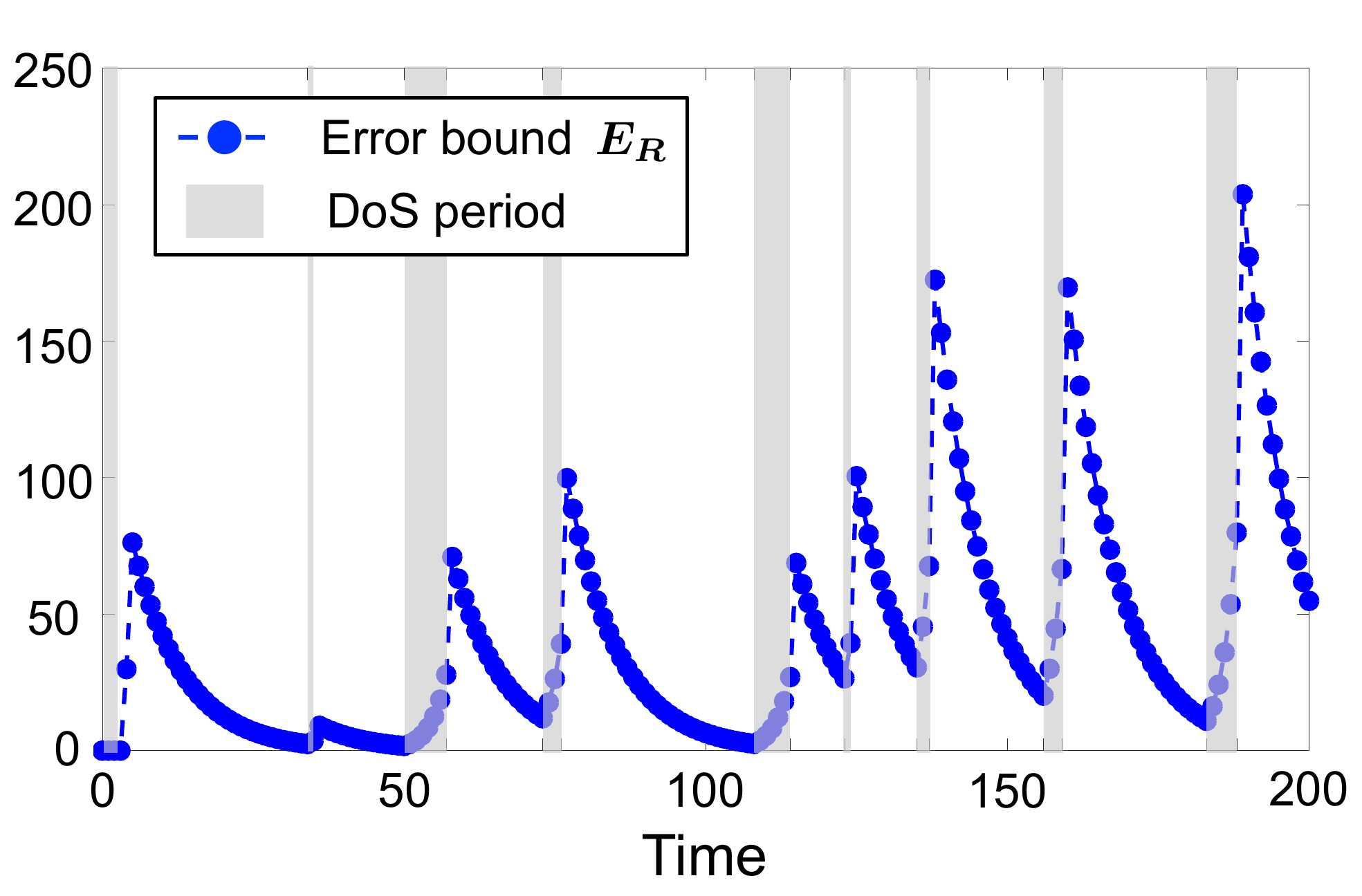}}
	\caption{Unstable case with frequency condition ($\nu_d = 0.15,~\nu_f = 0.045$). \label{fig:unstable_freq}}
	\vspace{-8pt}
\end{figure}

\section{Conclusion}
We proposed output encoding schemes resilient to DoS attacks and
obtained sufficient conditions on DoS duration and frequency bounds
for exponential convergence and Lyupunov stability with finite data rates.
The proposed encoding schemes are extensions of the zooming-in and zooming-out
method to the case with DoS.
Once an initial state bound is derived,
the coders decrease the quantization range in the absence of DoS.
However, if DoS attacks are detected, then
the coders increase the quantization range so that the output at the next time-step 
falls into the quantization region.
Moreover,
we discussed how to obtain state bounds under DoS attacks.
Future work
is to address more general networked control systems by considering
network phenomena at 
communication channels from the controller to the plant.

\appendix
\section{Proof of Theorem \ref{thm:prime}}
\renewcommand{\thetheorem}{\Alph{theorem}}

Let us first consider the case where $Q=1$, namely,
the matrix $A$ is similar to a diagonal matrix 
$\Lambda := {\rm diag} \left(\lambda e^{i 2\pi \frac{a_{1}}{\zeta}},\dots, \lambda e^{i 2\pi \frac{a_{n_x}}{\zeta}} \right)$,
where $\zeta \in \mathbb{N}$ be one or a prime number, $\lambda \in \mathbb{C}$ be nonzero, and
$a_{1},\dots,a_{n_x} \in \mathbb{Z}$ satisfy $a_{\ell_1} \not\equiv a_{\ell_2}$ (mod $\zeta$) for all
$\ell_1,\ell_2 = 1,\dots,n_x$ with $\ell_1 \not= \ell_2$.

To obtain a sufficient condition for the matrix $O\left( \{s_m\}_{m=0}^{\chi} \right)$
to be full column rank,
we use the following result on a generalized Vandermonde matrix:
\begin{lemma}[Theorem 6 of \cite{Evans1976}]
	\label{thm:evans}
	Let $\zeta$ be a prime number, and let $a_{1},\dots,a_{n} \in \mathbb{Z}$ and 
	$b_{1},\dots,b_{n} \in \mathbb{Z}$ satisfy $a_{\ell_1} \not\equiv a_{\ell_2}$ (mod $\zeta$) and $b_{\ell_1} \not\equiv b_{\ell_2}$ (mod $\zeta$) for all
	$\ell_1,\ell_2 = 1,\dots,n$ with $\ell_1 \not= \ell_2$. Then the generalized Vandermonde matrix
	\begin{equation*}
	\begin{bmatrix}
	e^{i2\pi \frac{a_1b_1}{\zeta}} & e^{i2\pi \frac{a_2b_1}{\zeta}} & \cdots & e^{i2\pi \frac{a_nb_1}{\zeta}} \\
	e^{i2\pi \frac{a_1b_2}{\zeta}} & e^{i2\pi \frac{a_2b_2}{\zeta}} & \cdots & e^{i2\pi \frac{a_nb_2}{\zeta}} \\
	\vdots & \vdots & \ddots & \vdots\\
	e^{i2\pi \frac{a_1b_n}{\zeta}} & e^{i2\pi \frac{a_2b_n}{\zeta}} & \cdots & e^{i2\pi \frac{a_nb_n}{\zeta}}  
	\end{bmatrix}
	\end{equation*}
	is invertible.
\end{lemma}

\begin{lemma}
	\label{lem:one_cycle}
	Assume that the matrix $A$ is similar to the diagonal matrix $\Lambda$ defined above.
	Suppose that  
	Assumptions \ref{assump:duration}, \ref{assump:observability}, \ref{assump:initial_time}, and \ref{assump:single_output} hold.
	If the DoS duration bound $\nu_d$ satisfies
	\begin{equation}
	\label{eq:duration_prime_onecycle}
	\nu_d  < \frac{\zeta-n_x+1}{\zeta},
	\end{equation}
	then $O\left( \{s_m\}_{m=0}^{\chi} \right)$ defined in \eqref{eq:Obs_matrix_DoS} 
	is full column rank
	by time $k=(\ell_e + 1)\zeta$, where $\ell_e \in \mathbb{Z}_+$ is the 
	maximum integer satisfying
	\begin{equation}
	\label{eq:ell_cond_prime_onecycle}
	\ell_e \leq \frac{\Pi_d}{\zeta-n_x+1-\zeta\nu_d }.
	\end{equation}
	\vspace{0pt}
\end{lemma}

\begin{proof}
	Let $s_0,s_1,\dots,s_{\chi}$ be the time-steps without DoS on the interval $[0,k)$, and
	define 
	\begin{equation}
	\label{eq:b_def}
	b_m := s_{m-1} - s_{0}\qquad \forall m=1,\dots,\chi+1.
	\end{equation}
	There exists an invertible matrix $R \in \mathbb{C}^{n_x \times n_x}$ such that 
	$AR = R \Lambda$. 
	Define 
	$C_\Lambda := CR = \begin{bmatrix}c_1 & \cdots & c_{n_x} \end{bmatrix}$ and
	\[ \mathcal{V} :=
	\begin{bmatrix}
	e^{i2\pi \frac{a_1b_1}{\zeta}} & e^{i2\pi \frac{a_2b_1}{\zeta}} & \cdots & e^{i2\pi \frac{a_{n_x}b_1}{\zeta}} \\
	e^{i2\pi \frac{a_1b_2}{\zeta}} & e^{i2\pi \frac{a_2b_2}{\zeta}} & \cdots & e^{i2\pi \frac{a_{n_x}b_2}{\zeta}} \\
	\vdots & \vdots & \ddots & \vdots\\
	e^{i2\pi \frac{a_1b_{\chi+1}}{\zeta}} & e^{i2\pi \frac{a_2b_{\chi+1}}{\zeta}} & \cdots & e^{i2\pi \frac{a_{n_x}b_{\chi+1}}{\zeta}}  
	\end{bmatrix}.
	\]
	Since $(C,A)$ is observable by Assumption \ref{assump:observability},
	it follows that $c_j \not=0$ for every $j=1,\dots,n_x$.
	We obtain
	\begin{align*}
	O\left( \{s_m\}_{m=0}^{\chi} \right) 
	&= 
	\begin{bmatrix}
	CA^{b_1} \\  \vdots \\ CA^{b_{\chi+1}} 
	\end{bmatrix} 
	=
	\begin{bmatrix}
	C_\Lambda \Lambda^{b_1} \\  \vdots \\ C_\Lambda \Lambda^{b_{\chi+1}} 
	\end{bmatrix}R^{-1} \\
	&=
	{\rm diag}(\lambda^{b_1},\dots,\lambda^{b_{\chi+1}}) \cdot
	\mathcal{V}
	\cdot
	{\rm diag}(c_1,\dots,c_{n_x}) R^{-1}.
	\end{align*}
	Therefore, the rank of 
	$O\left( \{s_m\}_{m=0}^{\chi} \right) $
	is equal to the rank 
	of 
	$
	\mathcal{V}
	$.
	By the assumption on $a_1,\dots,a_{n_x}$, 
	Lemma  \ref{thm:evans} shows that if there exist $\tilde b_1,\dots,\tilde b_{n_x} \in \{b_1,\dots,b_{\chi+1}\}$ such that
	$\tilde b_{\ell_1} \not\equiv \tilde b_{\ell_2}$ (mod $\zeta$) for all
	$\ell_1,\ell_2 = 1,\dots,n_x$ with $\ell_1 \not= \ell_2$,
	then 
	$
	\text{rank}~\!\mathcal{V}
	= n_x.
	$
	This implies that $O\left( \{s_m\}_{m=0}^{\chi} \right) $ is full column rank.
	
	Suppose, to reach a contradiction, that 
	for every $k = \ell \zeta $ with $\ell \in \mathbb{Z}_+$, there do not exist such
	\[
	\tilde b_1,\dots,\tilde b_{n_x}\in \{b_1,\dots,b_{\chi+1}\} = \{0, s_1 - s_0,\dots,s_{\chi} - s_0\},
	\]
	where $s_0,\dots,s_{\chi}$ are the time-steps without DoS on the interval $[0,k)$.
	Then DoS attacks occur at least $\zeta-n_x+1$ times during every 
	interval consisting of consecutive $\zeta$ time-steps.
	Hence
	$\Phi_d(\ell \zeta) \geq \ell(\zeta-n_x+1)$.
	By Assumption \ref{assump:duration}, we obtain
	$
	\Phi_d(\ell \zeta) \leq \Pi_d + \nu_d(\ell \zeta).
	$
	Therefore,
	\[
	\ell (\zeta-n_x+1) \leq \Pi_d + \nu_d (\ell \zeta), 
	\]
	which yields 
	$
	(\zeta-n_x+1 - \zeta \nu_d) \ell  \leq \Pi_d.
	$
	If the DoS duration bound $\nu_d$ satisfies \eqref{eq:duration_prime_onecycle},
	then we get a contradiction for $\ell \in \mathbb{Z}_+$ larger than the right side of \eqref{eq:ell_cond_prime_onecycle}.
	This implies that $O\left( \{s_m\}_{m=0}^{\chi} \right)$ 
	is full column rank by time $k=(\ell_e + 1) \zeta$, where
	$\ell_e \in \mathbb{Z}_+$ is the maximum integer satisfying \eqref{eq:ell_cond_prime_onecycle}.
	This completes the proof.
\end{proof}

Let us next consider the general case.
The following lemma provides a useful algebraic fact, which is used to show Theorem \ref{thm:prime}:
\begin{lemma}[Lemma 26 of \cite{Rohr2014}]
	\label{thm:multi_cyclic}
	Let $\lambda_1,\dots, \lambda_Q \in \mathbb{C}$ satisfy $\lambda_j \not=0$ for all $j=1,\dots,Q$
	and 
	$(\lambda_{j_1} / \lambda_{j_2})^k \not=1$ for every $j_1,j_2 = 1,\dots,Q$ with $j_1 \not=j_2$
	and every $k \in \mathbb{N}$.
	Let $w_1,\dots,w_Q \in \mathbb{C}$ satisfy $w_j\not=0$ for some $j =1,\dots,Q$.
	Then, there exist at most finitely many $k \in \mathbb{Z}_+$ such that 
	\[
	\sum_{j=1}^{Q}
	w_j\lambda_j^k = 0.
	\]
	\vspace{0pt}
\end{lemma}

We are now in a position to prove Theorem \ref{thm:prime}.

\begin{proof}[Proof of Theorem \ref{thm:prime}]
Let $s_0,s_1,\dots,s_{\chi}$ be the time-steps without DoS on the interval $[0,k)$, and
define $b_m \in \mathbb{Z}_+$ ($m=1,\dots,\chi+1$) as in \eqref{eq:b_def}.
There exists an invertible matrix $R \in \mathbb{C}^{n_x \times n_x}$ such that 
$AR = R\Lambda$. Define $C_\Lambda := CR = 
\begin{bmatrix}
C_1 & \cdots & C_Q
\end{bmatrix}$ 
with $C_j \in \mathbb{C}^{1 \times n_j}$ for every $j=1,\dots,Q$.
Since $(C,A)$ is observable by Assumption \ref{assump:observability},
it follows that $(C_j, \Lambda_j)$ is also observable for every $j=1,\dots,Q$.
Define
\begin{align*}
O_j &:= 
\begin{bmatrix}
C_j\Lambda_j^{b_1} \\ \vdots \\ C_j\Lambda_j^{b_{\chi+1}}
\end{bmatrix} \qquad \forall j = 1,\dots,Q.
\end{align*}
Then $O\left( \{s_m\}_{m=0}^{\chi} \right) = 
\begin{bmatrix}
O_1 & \cdots & O_{Q}
\end{bmatrix} R^{-1}$.

Assume that $k \geq \ell_{\min}\zeta$, where $\ell_{\min} \in \mathbb{Z}_+$
is the minimum integer satisfying
\[
\ell_{\min} > \frac{\Pi_d}{\zeta_j - n_j + 1 - \zeta_j \nu_d}\qquad \forall j = 1,\dots,Q.
\]
Lemma \ref{lem:one_cycle} shows that the matrix
$O_j$ is full column rank for every $j=1,\dots,Q$, because
\[
\frac{1}{\zeta} \leq \frac{\zeta_j-n_j+1}{\zeta_j}\qquad \forall j=1,\dots,Q.
\]
Assume further  
that 
there exists $v \in \mathbb{C}^{n_x}$ such that
\[
O\left( \{s_m\}_{m=0}^{\chi} \right) Rv=
\begin{bmatrix}
O_1 & \cdots & O_{Q}
\end{bmatrix} v
=0.
\]
To prove that $O\left( \{s_m\}_{m=0}^{\chi} \right)$ is full column rank, 
it is enough to show that $v = 0$. 
Partition $v$ as
$v = \begin{bmatrix}
v_1^* & \cdots & v_{Q}^*
\end{bmatrix}^*
$
with $v_j \in \mathbb{C}^{n_j}$ for every $j=1,\dots,Q$.
Suppose, to get a contradiction, $v_{j_0} \not= 0$ for some $j_0=1,\dots,Q$.
Since $O_{j_0}$ is full column rank, there exists $\tilde b \in \{b_1,\dots,b_{\chi+1}\}$ such that 
$C_{j_0}\Lambda_{j_0}^{\tilde b}v_{j_0} \not=0$.
Let $\tilde b \equiv \alpha$ (mod $\zeta$) with $0\leq  \alpha \leq \zeta-1$ and define 
$w_j := C_j (\Lambda_j/\lambda_j)^ \alpha v_j$ for each $j=1,\dots,Q$.
Note that $w_{j_0} \not = 0$.
Lemma \ref{thm:multi_cyclic} shows that 
there exists at most a finite number $Z_{j_0}$ of non-negative integers $k$ such that 
\[
\sum_{j=1}^{Q}
w_j\lambda_j^k = 0.
\]
Since $\Lambda_j^{\zeta} = \lambda_j^{\zeta} I_{n_j}$ for every $j=1,\dots,Q$,
\begin{align*}
C_\Lambda \Lambda^{\psi \zeta +  \alpha}v = 
\sum_{j=1}^Q
C_j\Lambda _j^{\psi \zeta +  \alpha}v_j 
=
\sum_{j=1}^Q
w_j \lambda_j^{\psi \zeta+ \alpha} \quad ~
\forall \psi \in \mathbb{Z}_+.
\end{align*}
Therefore,
if $\{b_1,\dots,b_{\chi+1}\}$ contains more than $Z_{j_0}$ elements in $\{\psi \zeta +  \alpha: \psi \in \mathbb{Z}_+\}$,
then 
\[
\begin{bmatrix}
C_\Lambda \Lambda^{b_1}v \\ \vdots \\
C_\Lambda \Lambda^{b_{\chi +1}}v
\end{bmatrix} =
\begin{bmatrix}
O_1 & \cdots & O_{Q}
\end{bmatrix} v = 0
\] 
contradicts the above fact obtained from Lemma \ref{thm:multi_cyclic}.
{\tiny }
From the discussion above,
it suffices to show that if the DoS duration bound $\nu_d$ satisfies
\eqref{eq:duration_prime}, then 
for every $\alpha=0,\dots,\zeta-1$ and every $Z \in \mathbb{Z}_+$, 
there exists $k \in \mathbb{Z}_+$ such that the set of time-steps $\leq k$ without DoS,
$\{s_0,\dots,s_{\chi}\}$, contains more than $Z$ elements in $\{\psi \zeta +  \alpha: \psi \in \mathbb{Z}_+\}$.
To this end, we assume by contradiction that for every $k = \ell \zeta$ with $\ell \in \mathbb{Z}_+$, 
the number of elements in $\{s_0,\dots,s_{\chi}\} \cup \{\psi \zeta + \alpha: \psi \in \mathbb{Z}_+\}$
does not exceed $Z$.
Then $\Phi_d(\ell \zeta) \geq \ell - Z$.
By  Assumption \ref{assump:duration}, $\Phi_d(\ell \zeta) \leq \Pi_d+\nu_d(\ell \zeta)$.
We obtain
\[
\ell - Z \leq \Pi_d+\nu_d(\ell \zeta)
\]
and hence
$
(1 - \zeta \nu_d) \ell \leq \Pi_d + Z.
$
By \eqref{eq:duration_prime}, 
\begin{equation}
\label{eq:prime_for_proof}
\ell \leq \frac{\Pi_d + Z}{1 - \zeta \nu_d},
\end{equation}
which contradicts for a sufficiently large $\ell \in \mathbb{Z}_+$.
Moreover, 	$\{s_0,\dots,s_{\chi}\}$ contains more than $Z$ elements in $\{\psi \zeta +  \alpha: \psi \in \mathbb{Z}_+\}$
for 
$k=(\ell_e + 1) \zeta$, where
$\ell_e \in \mathbb{Z}_+$ is the maximum integer that does not exceed the right side of \eqref{eq:prime_for_proof}.
This completes the proof.
\end{proof}


\bibliographystyle{siamplain}

\end{document}

%% file: ex_shared.tex
\usepackage{lipsum}
\usepackage{amsfonts}
\usepackage{graphicx}
\usepackage{epstopdf}
\usepackage{algorithmic}
\ifpdf
  \DeclareGraphicsExtensions{.eps,.pdf,.png,.jpg}
\else
  \DeclareGraphicsExtensions{.eps}
\fi

\numberwithin{theorem}{section}

\newcommand{\TheTitle}{Stabilization of Networked Control Systems under DoS Attacks and Output Quantization} 
\newcommand{\TheAuthors}{Masashi Wakaiki, Ahmet Cetinkaya,
	and Hideaki Ishii}

\headers{Stabilization under DoS Attacks and Quantization}{\TheAuthors}

\title{{\TheTitle}\thanks{
		This paper was partially presented at the American Control Conference 2018 \cite{Wakaiki2018ACC}.
		\funding{
			This work was supported in part by JSPS KAKENHI Grant Numbers
			JP17K14699 and JP18H01460 and by the JST CREST Grant No. JPMJCR15K3.}}}

\author{
  Masashi Wakaiki\thanks{Graduate School of System Informatics, Kobe University,  Hyogo 657-8501, Japan
    (\email{wakaiki@ruby.kobe.u.ac.jp}).}
  \and
  Ahmet Cetinkaya\thanks{Department of Computer Science, Tokyo Institute of Technology, Yokohama, 226-8502, Japan
  	(\email{ahmet@sc.dis.titech.ac.jp; ishii@c.titech.ac.jp}).}
  \and 
  Hideaki Ishii\footnotemark[3]
}

\usepackage{amsopn}


%% file: ex_article.bbl
\begin{thebibliography}{10}
	
	\bibitem{Amin2009}
	{\sc S.~Amin, A.~A. C\'ardenas, and S.~S. Sastry}, {\em Safe and secure
		networked control systems under denial-of-service attacks}, in Proc. 12th
	HSCC, 2009.
	
	\bibitem{Awerbuch2008}
	{\sc B.~Awerbuch, R.~Curtmola, D.~Holmer, C.~Nita-Rotaru, and H.~Rubens}, {\em
		{ODSBR: An on-demand secure Byzantine resilient routing protocol for wireless
			ad hoc networks}}, ACM Trans. Inf. and System Security, 10, Article No. 6
	(2008).
	
	\bibitem{Bhattacharya2013}
	{\sc S.~Bhattacharya, A.~Gupta, and T.~Ba\c{s}ar}, {\em {Jamming in mobile
			networks: A game-theoretic approach}}, J. Numer. Algeb. Control Optim., 3
	(2013), pp.~1--30.
	
	\bibitem{Brockett2000}
	{\sc R.~W. Brockett and D.~Liberzon}, {\em Quantized feedback stabilization of
		linear systems}, IEEE Trans. Automat. Control, 45 (2000), pp.~1279--1289.
	
	\bibitem{Cetinkaya2017}
	{\sc A.~Cetinkaya, H.~Ishii, and T.~Hayakawa}, {\em Networked control under
		random and malicious packet losses}, IEEE Trans. Automat. Control, 62 (2017),
	pp.~2434--2449.
	
	\bibitem{Cetinkaya2018}
	{\sc A.~Cetinkaya, H.~Ishii, and T.~Hayakawa}, {\em Analysis of stochastic
		switched systems with application to networked control under jamming
		attacks}.
	\newblock To appear in {\em IEEE Trans. Automat. Control}, 2018.
	
	\bibitem{Ahmet2018arXiv}
	{\sc A.~Cetinkaya, H.~Ishii, and T.~Hayakawa}, {\em A probabilistic
		characterization of random and malicious communication failures in multi-hop
		networked control}.
	\newblock To appear in {\it SIAM J. Control Optim.}, 2018,
	\url{https://arxiv.org/pdf/1711.06855.pdf}.
	
	\bibitem{Checkoway2011}
	{\sc S.~Checkoway, D.~McCoy, B.~Kantor, D.~Anderson, H.~Shacham, S.~Savage,
		K.~Kocher, A.~Czeskis, F.~Roesner, and T.~Kohno}, {\em Comprehensive
		experimental analyses of automotive attack surfaces}, in Proc. USENIX
	Security Symposium, 2011.
	
	\bibitem{Chen2018}
	{\sc X.~Chen, Y.~Wang, and S.~Hu}, {\em Event-based robust stabilization of
		uncertain networked control systems under quantization and denial-of-service
		attacks}, Inf. Sci., 459 (2018), pp.~369--386.
	
	\bibitem{Chong2015}
	{\sc M.~S. Chong, M.~Wakaiki, and J.~P. Hespanha}, {\em Observability of linear
		systems under adversarial attacks}, in Proc. ACC'15, 2015.
	
	\bibitem{Persis2015}
	{\sc C.~De~Persis and P.~Tesi}, {\em {Input-to-state stabilizing control under
			denial-of-service}}, IEEE Trans. Automat. Control, 60 (2015), pp.~2930--2944.
	
	\bibitem{Persis2016}
	{\sc C.~De~Persis and P.~Tesi}, {\em Networked control of nonlinear systems
		under denial-of-service}, Systems \& Control Letters, 96 (2016),
	pp.~124--131.
	
	\bibitem{Ding2017}
	{\sc K.~Ding, Y.~Li, D.~E. Quevedo, S.~Dey, and L.~Shi}, {\em {A multi-channel
			transmission schedule for remote state estimation under DoS attacks}},
	Automatica, 78 (2017), pp.~194--201.
	
	\bibitem{Evans1976}
	{\sc R.~J. Evans and I.~M. Isaacs}, {\em {Generalized Vandermonde determinants
			and roots of unity of prime order}}, Proc. Amer. Math. Soc., 58 (1976),
	pp.~51--54.
	
	\bibitem{Fawzi2014}
	{\sc H.~Fawzi, P.~Tabuada, and S.~Diggavi}, {\em {Secure estimation and control
			for cyber-physical systems under adversarial attacks}}, IEEE Trans. Automat.
	Control, 59 (2014), pp.~1454--1467.
	
	\bibitem{Feng2017}
	{\sc S.~Feng and P.~Tesi}, {\em Resilient control under denial-of-service:
		Robust design}, Automatica, 79 (2017), pp.~42--51.
	
	\bibitem{Feng2018}
	{\sc S.~Feng, P.~Tesi, C.~De~Persis, A.~Cetinkaya, and H.~Ishii}, {\em Data
		rates for stabilizing control under denial-of-service attacks}.
	\newblock 2018, under preparation.
	
	\bibitem{Hespanha2007}
	{\sc J.~P. Hespanha, P.~Naghshtabrizi, and Y.~Xu}, {\em {A survey of recent
			results in networked control systems}}, Proc. IEEE, 95 (2007), pp.~138--162.
	
	\bibitem{Imer2006}
	{\sc O.~Imer, S.~Y\"{u}ksel, and T.~Ba\c{s}ar}, {\em {Optimal control of LTI
			systems over unreliable communication links}}, Automatica, 42 (2006),
	pp.~1429--1439.
	
	\bibitem{Ishii2012}
	{\sc H.~Ishii and K.~Tsumura}, {\em {Data rate limitations in feedback control
			over networks}}, IEICE Trans. Fundamentals, E95-A (2012), pp.~680--690.
	
	\bibitem{Jungers2018}
	{\sc R.~M. Jungers, A.~Kundu, and W.~P. M.~H. Heemels}, {\em Observability and
		controllability analysis of linear systems subject to packet losses}.
	\newblock To appear in {\em IEEE Trans. Automat. Control}, 2018.
	
	\bibitem{Kerns2014}
	{\sc A.~J. Kerns, D.~P. Shepard, J.~A. Bhatti, and T.~E. Humphreys}, {\em
		{Unmanned aircraft capture and control via GPS spoofing}}, J. Field Robot.,
	31 (2014), pp.~617--636.
	
	\bibitem{Kikuchi2017}
	{\sc K.~Kikuchi, A.~Cetinkaya, T.~Hayakawa, and H.~Ishii}, {\em Stochastic
		communication protocols for multi-agent consensus under jamming attacks}, in
	Proc. 56th IEEE CDC, 2017.
	
	\bibitem{Liberzon2003Automatica}
	{\sc D.~Liberzon}, {\em Hybrid feedback stabilization of systems with quantized
		signals}, Automatica, 39 (2003), pp.~1543--1554.
	
	\bibitem{Liberzon2003}
	{\sc D.~Liberzon}, {\em On stabilization of linear systems with limited
		information}, IEEE Trans. Automat. Control, 48 (2003), pp.~304--307.
	
	\bibitem{Liberzon2014}
	{\sc D.~Liberzon}, {\em Finite data-rate feedback stabilization of switched and
		hybrid linear systems}, Automatica, 50 (2014), pp.~409--420.
	
	\bibitem{Liberzon2005}
	{\sc D.~Liberzon and J.~P. Hespanha}, {\em {Stabilization of nonlinear systems
			with limited information feedback}}, IEEE Trans. Automat. Control, 50 (2005),
	pp.~910--915.
	
	\bibitem{Liu2014JA}
	{\sc S.~Liu, P.~X. Liu, and A.~E. Saddik}, {\em A stochastic game approach to
		the security issue of networked control systems under jamming attacks}, J.
	Frankl. Inst., 351 (2014), pp.~4570--4583.
	
	\bibitem{Lu2018}
	{\sc A.-Y. Lu and G.-H. Yang}, {\em Input-to-state stabilizing control for
		cyber-physical systems with multiple transmission channels under denial of
		service}, IEEE Trans. Automat. Control, 63 (2018), pp.~1813--1820.
	
	\bibitem{Mo2009}
	{\sc Y.~Mo and B.~Sinopoli}, {\em Secure control against replay attacks}, in
	Proc. Allerton Conf. on Communications, Control and Computing, 2009.
	
	\bibitem{Nair2007}
	{\sc G.~N. Nair, F.~Fagnani, S.~Zampieri, and R.~J. Evans}, {\em Feedback
		control under data rate constraints: An overview}, Proc. IEEE, 95 (2007),
	pp.~108--137.
	
	\bibitem{Park2011}
	{\sc S.~Y. Park and A.~Sahai}, {\em {Intermittent Kalman filtering: eigenvalue
			cycles and nonuniform sampling}}, in Proc. ACC 2011,
	\url{https://arxiv.org/abs/1308.6107}.
	
	\bibitem{Pelechrinis2011}
	{\sc K.~Pelechrinis, M.~Iliofotou, and S.~V. Krishnamurty}, {\em {Denial of
			service attacks in wireless networks: The case of jammers}}, IEEE Commun.
	Surveys Tuts., 13 (2011), pp.~245--257.
	
	\bibitem{Rohr2014}
	{\sc E.~R. Rohr, D.~Marelli, and M.~Fu}, {\em Kalman filtering with
		intermittent observations: On the boundedness of the expected error
		covariance}, IEEE Trans. Automat. Control, 59 (2014), pp.~2724--2738.
	
	\bibitem{Rosenbrock1972}
	{\sc H.~H. Rosenbrock}, {\em Computer-Aided Control System Design}, New York:
	Academic Press, 1974.
	
	\bibitem{Senejohnny2017}
	{\sc D.~Senejohnny, P.~Tesi, and C.~De~Persis}, {\em A jamming-resilient
		algorithm for self-triggered network coordination}.
	\newblock To appear in {\em IEEE Trans. Control Network Systems}, 2017.
	
	\bibitem{Shoukry2016}
	{\sc Y.~Shoukry and P.~Tabuada}, {\em Event-triggered state observers for
		sparse sensor noise/attacks}, IEEE Trans. Automat. Control, 61 (2016),
	pp.~2079--2091.
	
	\bibitem{Tsumura2009}
	{\sc K.~Tsumura, H.~Ishii, and H.~Hoshina}, {\em Tradeoffs between quantization
		and packet loss in networked control of linear systems}, Automatica, 45
	(2009), pp.~2963--2970.
	
	\bibitem{Wakaiki2018ACC}
	{\sc M.~Wakaiki, A.~Cetinkaya, and H.~Ishii}, {\em {Quantized output feedback
			stabilization under DoS attacks}}, in Proc. ACC 2018, 2018.
	
	\bibitem{Wakaiki2017TAC}
	{\sc M.~Wakaiki and Y.~Yamamoto}, {\em Stabilization of switched linear systems
		with quantized output and switching delays}, IEEE Trans. Automat. Control, 62
	(2017), pp.~2958--2964.
	
	\bibitem{Wakaiki2017IFAC}
	{\sc M.~Wakaiki, T.~Zanma, and K.-Z. Liu}, {\em {Quantized output feedback
			stabilization by Luenberger observers}}, in Proc. 20th IFAC WC, 2017.
	\newblock {arXiv:1703.06567}.
	
	\bibitem{Wang2011}
	{\sc L.~Y. Wang, C.~Li, G.~G. Yin, L.~Guo, and C.-Z. Xu}, {\em State
		observability and observers of linear-time-invariant systems under irregular
		sampling and sensor limitations}, IEEE Trans. Automat. Control, 56 (2011),
	pp.~2639--2654.
	
	\bibitem{Zeng2017}
	{\sc S.~Zeng, H.~Ishii, and F.~Allgower}, {\em Sampled observability and state
		estimation of linear discrete ensembles}, IEEE Trans. Automat. Control, 62
	(2017), pp.~2406--2418.
	
	\bibitem{Zhang2013}
	{\sc L.~Zhang, H.~Gao, and O.~Kaynak}, {\em {Network-induced constraints in
			networked control systems--A survey}}, IEEE Trans. Ind. Inform., 9 (2013),
	pp.~406--416.
	
	\bibitem{Zhu2014}
	{\sc M.~Zhu and S.~Mart\'inez}, {\em On the performance analysis of resilient
		networked control systems under replay attacks}, IEEE Trans. Automat.
	Control, 59 (2014), pp.~804--808.
	
\end{thebibliography}
